\documentclass[aps,reprint,floatfix]{revtex4-1}
\usepackage{amsmath,amssymb,graphicx,mathtools}
\usepackage{float}

\usepackage[T1]{fontenc}

\usepackage{colortbl}
\usepackage{multirow}

\usepackage[ruled]{algorithm2e}

\makeatletter
\DeclareRobustCommand{\cev}[1]{%
  \mathpalette\do@cev{#1}%
}

\makeatother

\newcounter{theoremcounter}
\stepcounter{theoremcounter}

\newtheorem{theorem}{Theorem}

\newtheorem{lemma}{Lemma}

\newtheorem{corollary}{Corollary}

\newenvironment{proof}[1][Proof]{\begin{trivlist}
\item[\hskip \labelsep {\bfseries #1}]}{\end{trivlist}}

\newcommand{\qed}{\hfill $\blacksquare$}

\graphicspath{./}

\newcommand{\bra}[1]{\langle #1|}
\newcommand{\ket}[1]{|#1\rangle}
\newcommand{\braket}[2]{\langle #1|#2\rangle}
\newcommand{\expval}[3]{\langle #1|#2|#3\rangle}
\newcommand{\ketbra}[2]{\ket{#1}\bra{#2}}

\DeclareMathOperator{\Tr}{Tr}

\let\Re\relax

\DeclareMathOperator{\Re}{\text{Re}}

\makeatletter
\newcommand{\vast}{\bBigg@{4}}
\newcommand{\Vast}{\bBigg@{5}}
\makeatother

\begin{document}

\title{More Optimal Simulation of Universal Quantum Computers}
\author{Lucas Kocia}
\affiliation{Sandia National Laboratories, Livermore, California 94550, U.S.A.}
\author{Genele Tulloch}
\affiliation{North Carolina Central University, Durham, North Carolina 27707, U.S.A.}
\begin{abstract}
Validating whether a quantum device confers a computational advantage often requires classical simulation of its outcomes. The worst-case sampling cost of \(L_1\)-norm based simulation has plateaued at \(\le(2+\sqrt{2})\xi_t \delta^{-1}\) in the limit that \(t \rightarrow \infty\), where \(\delta\) is the additive error and \(\xi_t\) is the stabilizer extent of a \(t\)-qubit magic state. We reduce this prefactor \(68\)-fold by a leading-order reduction in \(t\) through correlated sampling. The result exceeds even the average-case of the prior state-of-the-art and current simulators accurate to multiplicative error. Numerical demonstrations support our proofs. The technique can be applied broadly to reduce the cost of \(L_1\) minimization.
\end{abstract}
\maketitle

Quantum computers hold great promise for solving many key problems of interest that do not scale favorably on classical computers\cite{Lloyd96,Grover96,Shor99}. On the way to fault-tolerant implementations~\cite{Shor96,Steane99,Raussendorf07}, it is vital to establish thresholds for validating demonstrations of quantum advantage~\cite{Ried15,Riste17,Bravyi18,Arute19,Bravyi20,Zhong20,Haferkamp20,Maslov21,Wu21_2}. This often requires directly simulating the outcomes of a noisy quantum device on classical computers~\cite{Harrow17,Terhal18,Preskill18}, which is widely considered to scale exponentially asymptotically with a universal resource~\cite{Lloyd95,Divincenzo95,Knill04,Bravyi05}.

A great deal of recent work on classical simulation has focused on optimizing tensor-network methods~\cite{Markov08,Pednault19,Huang20_3,Pan20,Zhou20,Gray21,Pan21,Liu21,Feng22}, which scale w.r.t.~the treewidth of a graph corresponding to a circuit's connectivity. There is also a push for \(L_1\)-norm based sampling~\cite{Bravyi16_1,Howard18,Seddon20,Pashayan21}, which scale w.r.t.~the number of magic states when a circuit is reexpressed as a measurement-based computation~\cite{Raussendorf01}. These are complementary efforts; contracted graph states produced from the former can serve as the input for \(L_1\)-norm based outcome estimation~\cite{Markov08}. It is important to reduce the cost of both steps as near-term quantum devices grow. These classical algorithms are frequently called weak simulators, since they are accurate to additive error in contrast to multiplicative error (i.e. strong simulation).

Frequently, few-qubit tensored magic states are sufficient for efficient decomposition of universal circuits~\cite{Trout15,Campbell17}. Classical simulators often decompose magic states into stabilizer states with minimal \(L_1\) norm because their inner products can be calculated particularly efficiently~\cite{Aaronson04,Bravyi16_1,Bravyi16_2,Howard18}. These methods use independent and identically distributed (i.i.d.)~sampling of stabilizer states. Early algorithms required \(\mathcal O(\xi_t \delta^{-2})\) samples~\cite{Bravyi16_1,Bravyi16_2,Howard18}, where \(\delta\) is the additive error and \(\xi_t\) is the stabilizer extent of a \(t\)-qubit magic state (e.g.~\(\xi_t = 2^{\sim 0.228 t}\) for the \(T\) gate). Their requirements were reduced to \((2+\sqrt{2})\xi_t \delta^{-1}\) samples~\cite{Seddon20}, which is the current state-of-the-art. Asymptotic worst-case reductions have since abated leaving open whether further reduction is possible. 

We will show how supplementing i.i.d.~samples with correlated samples produces a leading-order reduction in \(t\) that greatly lowers the asymptotic prefactor of weak simulation. This increases the size of universal quantum circuits that are simulatable by classical computers by an asymptotic factor of \({\sim} 68\). Figure~\ref{fig:scaling_improvement} summarizes how this reduction outperforms the prior state-of-the-art for the \(T\) gate, and even outperforms current strong simulators.

\begin{figure}[H]
  \includegraphics[scale=0.28]{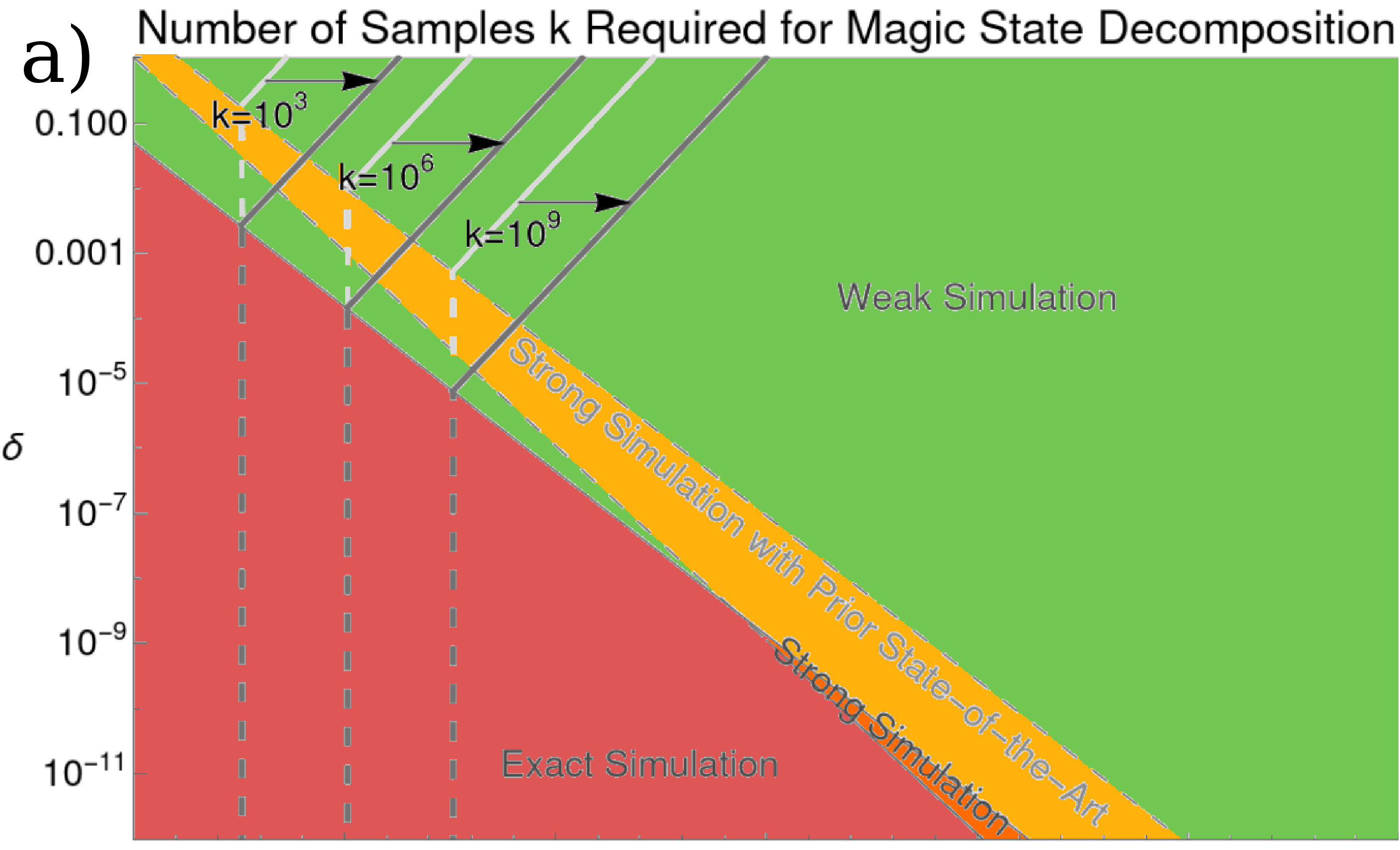}
  \includegraphics[scale=0.284]{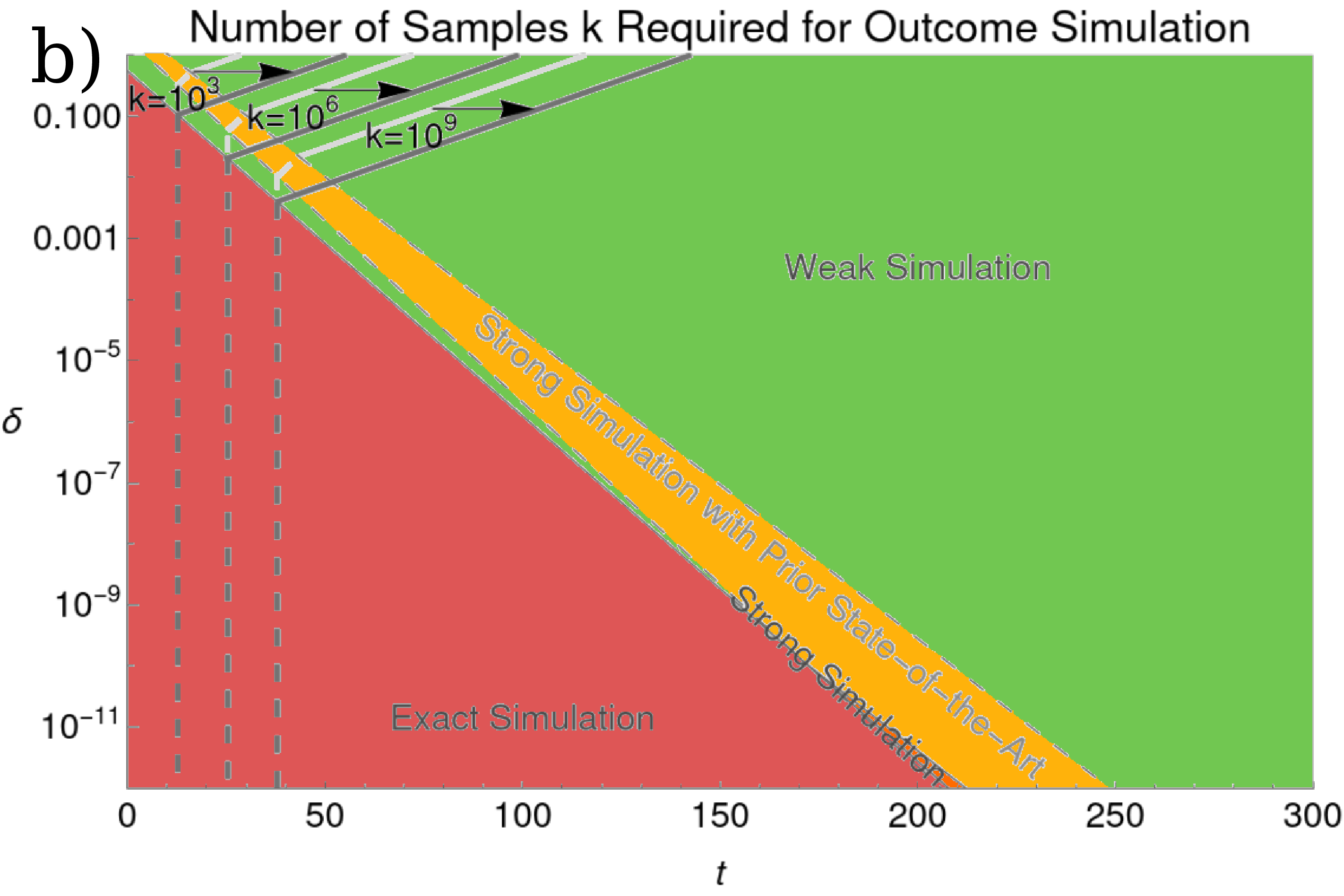}
  \caption{A plot of classical simulation cost w.r.t.~\(t\) and \(\delta\) showing where it scales optimally under the leading weak, strong and exact simulation algorithms in terms of a) decomposing \(T\) gate magic states into stabilizer states and b) estimating outcomes from this decomposition. The grey curves correspond to simulation with a constant number of stabilizer states and show the reduction by \(t\approx 27\) magic states by the new algorithm (a \(68\)-fold decrease). Though these plots assume asymptotic scaling for all algorithms at all \(t\) and \(\delta\), numerical results demonstrate that they are nearly indiscernible from more accurate plots (see Figure~\ref{fig:numerics} and Figure~\ref{fig:numerics2}).}
\label{fig:scaling_improvement}
\end{figure}

\section{Weak Simulation Exposition}
\label{exposition}

Quantum computers output outcomes that are effectively correct only up to an additive error. Any interactive test using polynomial resources finds quantum devices efficiently indistinguishable even when their outputs have additive error~\cite{Aaronson11,Aaronson14,Fefferman15,Morimae17,Pashayan20}. This greatly simplifies their classical simulation compared to systems that are accurate to multiplicative error. In particular, weak simulation should be more efficient in the worst-case when errors are not too small since its complexity class, BQP-complete, is contained in strong simulation's complexity class (\#P-hard)~\cite{Pashayan20}.

One way to see how additive error constraints simplify problems is through the way that they reduce simulation optimization to an \(L_1\) minimization~\cite{Bravyi16_1}. \(L_1\) minimizations, which favor relaxed (convex) sparse support, are known to scale more favorably than similar \(L_0\) minimizations~\cite{Natarajan95,Ge11}, which more strongly favor sparse support and are more relevant for multiplicative error. Given any state decomposed into an orthonormal basis, \(\ket{\Psi} = \sum_i c_i \ket{\phi_i}\), the \(L_1\) norm of the state is the sum of the absolute values of its coefficients, \(\|\ket{\Psi}\|_1 \equiv \sum_i |c_i|\). This is a multiplicative quantity, \(\|\ket{\Psi}^{\otimes t}\|_1 = \|\ket{\Psi}\|_1^t\), and is basis-independent.

Here we will be interested in stabilizer state decompositions of given states \(\ket \Psi\). Stabilizer states are generally neither orthogonal nor are their decompositions of states unique~\cite{Gottesman98}. Therefore, since norms are basis-independent, we cannot strictly associate their decompositions with an \(L_1\) norm. Instead, we refer to the \emph{stabilizer extent}, \(\xi\): the minimum sum squared of absolute values of the coefficients of stabilizer states decomposing a given state, minimized over all such decompositions~\cite{Howard18}.

The stabilizer extent shares an important property with true minimal \(L_1\) norms: it is multiplicative, though for only one-, two-, and three-qubit tensored states~\cite{Howard18,Heimendahl21}. Often few-qubit tensored magic states are sufficient for efficient decomposition of universal circuits~\cite{Trout15,Campbell17}, so this multiplicativity property can be usefully applied.

For this reason and for easier exposition, we initially concern ourselves with \(t\)-tensored one-qubit magic states, \(\ket{\Psi_\phi}^{\otimes t}\). These can be parametrized by their angle \(\phi\) w.r.t.~the top pole of the Bloch sphere,
\begin{equation}
\label{eq:diagstate}
\ket{\Psi_\phi}^{\otimes t} \equiv (2 \nu)^{-t} \sum_{\substack{x \equiv x_1\cdots x_t\\x_i \in \{\tilde 0,\tilde 1\}}} \ket{x_1} \otimes \cdots \otimes \ket{x_t},
\end{equation}
where 
\begin{eqnarray}
\ket{\tilde 0} \equiv \frac{i}{\sqrt{2}} (-i + e^{-\pi i/4}) (-i + e^{i \phi}) \ket{0},\\
\ket{\tilde 1} \equiv \frac{i}{\sqrt{2}} (1 + e^{-\pi i/4}) (1 - e^{i \phi}) \frac{1}{\sqrt{2}} (\ket{0} + \ket{1}),
\end{eqnarray}
and \(\nu \equiv \cos \pi/8\). The stabilizer decomposition into \(\ket{\tilde 0}\) and \(\ket{\tilde 1}\) saturates \(\ket{\Psi_\phi}^{\otimes t}\)'s stabilizer extent,
\begin{equation}
\left(\sum_{x \in \{\tilde 0, \tilde 1\}^t} |c_{x}|\right)^2 = \left(\sqrt{1 - \sin\phi} + \sqrt{1 - \cos\phi}\right)^{2t} \equiv \xi_t(\phi).
\end{equation}
As a result, it is the most efficient stabilizer state decomposition to use when representing \(\ket{\Psi_\phi}^{\otimes t}\) up to some additive error \(\delta\) using \(L_1\)-norm sampling techniques.

\section{Sparsification}
\label{sparsification}

Additive error not only simplifies things by picking out the most efficient stabilizer decomposition, it also permits us to sparsify the stabilizer state decomposition of \(\ket{\Psi_\phi}^{\otimes t}\) given by Eq.~\ref{eq:diagstate} to produce an approximation, \(\ket{\psi}\). We only need to keep enough \(\ket{x}\) states such that \(\|\ket{\psi} - \ket{\Psi_\phi}^{\otimes t}\|_1 \le \delta\).

\begin{figure}[h]
\input{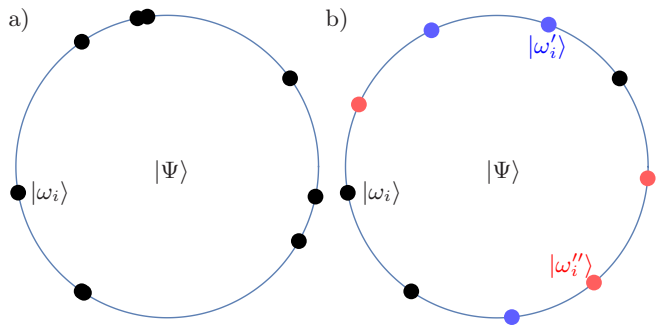}
\caption{Sketch of the different ensembles produced by (a) i.i.d.~samples and (b) correlated \(L_1\) samples. The expectation value of the i.i.d.~stabilizer state samples, \(\{\ket{\omega_i}\}_i\), will approach the desired state \(\ket{\Psi}\). However, any finite set will on average deviate from this expectation value. This ensemble can be transformed to a correlated one by supplementing a subset with mutually dissimilar states, \(\ket{\omega_i'}\) and \(\ket{\omega_i''}\), and discarding the rest. The expectation value of any such correlated finite set will be closer to \(\ket{\Psi}\).}
\label{fig:sketch}
\end{figure}

A simple way for selecting which \(\ket{x}\) to keep, is to treat the set \(\{\ket{x}\}\) as an ensemble defined by an i.i.d.~random variable. Given an ensemble of \(k\) i.i.d.~stabilizer states \(\{\ket{x}\}\) such that \(\mathbb E (\ket{x}) = \|c\|_1^{-2} \ket{\Psi_\phi}^{\otimes t}\), it follows that~\cite{Bravyi16_1, Bravyi16_2}
\begin{eqnarray}
  \label{eq:expval}
  \mathbb E (\| \ket{\Psi_\phi}^{\otimes t} - \ket{\psi}\|^2) &\le& \frac{\xi_t(\phi)}{k} - \frac{\gamma}{k},
\end{eqnarray}
where \(\gamma\) defines the contribution from the inner product cross terms. For i.i.d.~sampling, these cancel out and \(\gamma = 1\). The number of i.i.d.~samples must be \(k \in \mathcal O(\xi_t \delta^{-2})\) for this expectation value to be less than or equal to \(\delta^2\).

If i.i.d.~sampling is supplemented with correlated sampling, it is possible to obtain this expectation value with fewer samples. In particular, by supplementing each independently sampled state in \(k/f_t\) subsets of the ensemble with \(f_t\) other stabilizer states and discarding the rest, \(\gamma\) becomes
\begin{equation}
  \label{eq:gamma}
 \gamma = 1 + f_t - \sum_j^{f_t} \xi_t(\phi) \mathbb E(\braket{x}{y(x;j)}),
\end{equation}
where \(y(x;j)\) maps \(x\) to the \(f_t\) bitstrings corresponding to the states \(\{\ket{y(x;j)}\}_{j=1}^{f_t}\) that \(\ket{x}\) is correlated with. If the third term in Eq.~\ref{eq:gamma} can become negligible, then \(\gamma\) will decrease the expectation value in Eq.~\ref{eq:expval} by \(f_t\). This decreases the number of samples required to satisfy Eq.~\ref{eq:expval} down to \(k \in \mathcal O((\xi_t-f_t)\delta^{-2})\).

For the third term in Eq.~\ref{eq:gamma} to become asymptotically negligible compared to \(1{+}f_t\), the average inner product between correlated supplemental states must be upper bounded by the inverse stabilizer extent,
\begin{equation}
  \label{eq:constraint}
\mathbb E(\braket{x}{y(x;j)}) \in \mathcal O(\xi_t(\phi)^{-1}).  
\end{equation}
This means that the supplemental states should have the smallest possible inner product with each other and be as numerous as possible (i.e. be maximally mutually ``dissimilar'' from each other).
A sketch of the phenomenon can be found in Figure~\ref{fig:sketch}. With such a supplemental set of correlated states, the expectation value of the ensemble converge more quickly to \(\ket{\Psi_\phi}^{\otimes t}\) on average.

To choose optimal correlated sets of supplemental stabilizer states, we can distinguish between candidate stabilizer state subsets by their different average inner products. For one-qubit magic states, it is useful to take advantage of the property that these are incurred by mutual bitflips among correlated \(\ket{x}\) and \(\ket{y(x;j)}\) states. The number of bits, \(\alpha t\), that supplemental \(t\)-bitstrings must differ by can be determined by bounding the \(\alpha t\) product of bitflipped one-qubit states with Eq.~\ref{eq:constraint},
\begin{align}
\braket{\tilde 0}{\tilde 1}^{\alpha t} =& \left((2^{-\frac{1}{2}} +1) (\sin \phi +\cos \phi -1) \right)^{\alpha t} \le \xi^{-1}_t(\phi),
\end{align}
which is satisfied when
\begin{equation}
  \label{eq:bitflip_constraint}
  \alpha \le \frac{2 \left(2 \log 2+\log \left(1-\frac{1}{\sqrt{2}}\right)\right)}{\log 2} \approx 0.46,
\end{equation}
corresponding to mutually differing by approximately \(t/2\) bits on average. 

It is possible to generate a supplemental set of \(2t-1\) \(t\)-bitstrings that satisfy Eq.~\ref{eq:bitflip_constraint} with no asymptotic overhead:
\setcounter{lemma}{0}
\begin{lemma}[\(2t-1\) Generation]
  \label{le:linear_generation}
Given a \(t\)-bitstring for \(t\) a power of \(2\), there exist \(2t-1\) additional bitstrings that mutually differ from each other and the given bitstring by at least \(t/2\) bitflips and can be found in subquadratic time.
\end{lemma}
See Appendix~\ref{app:linear_scaling} for the constructive proof.

An algorithm that generates \(2t{-}1\) additional bitstrings that differ by at least \(t/2\) bitflips (such as those given in Table~\(1\) in Appendix~\ref{app:linear_scaling}), is given in Algorithm~\ref{alg:2tbitflips}. 

\begin{algorithm}[H]
  \KwData{$k$ such that $t=2^{k}$.}
  \KwResult{\emph{bitstring} array.}
  \Begin{
    \emph{yis} \(\leftarrow\) \(\{\alpha = 10, \beta = 01, \gamma = 00\}\)\;
    \For{treedepth \(\leftarrow 1\) \KwTo \(k-1\)}{
      \emph{newyis} \(\leftarrow\) \(\{\}\)\;
      \For{i \(\leftarrow 1\) \KwTo \(length(yis)\)}{
        \emph{newyis}[$i$] \(\leftarrow\) clone \emph{yis}[$i$]\;
        \emph{newyis}[$i+length(yis)$] \(\leftarrow\) complement lower half bits of \emph{newyis}[i]\;
      }
      \emph{newyis}[$(2\times length(yis)+1)$] \(\leftarrow\) complement upper half bits of \emph{newyis}[$3$]\;
      \emph{yis} \(\leftarrow\) \emph{newyis}\;
    }
  }
  \caption{Generate additional bitstrings that differ from the \(t\)-bitstring of all \(1\)s by at least \(t/2\) bitflips.}
  \label{alg:2tbitflips}
\end{algorithm}

The outer loop is over \(log(t)-1\) elements and the inner loop is over \(2(t/2)-1\) elements that involve operations on \(t\) bits, together producing \(\mathcal O(t^2 \log t)\) elements. Therefore, the runtime is \(\mathcal O(t^2 \log t)\) Since this runtime generates \(2t-1\) bitstrings, the overall scaling should be divided by this number. This means that the overall runtime is \(\mathcal O(t \log t)\) to generate each supplemental bitstring. This is confirmed by numerical evaluation as depicted in Figure~\ref{fig:2tgenerationtiming}.

Since this cost is added to the overall cost of weak simulation, which is exponential in \(t\), it is asymptotically negligible.

\begin{figure}[H]
\includegraphics[scale=0.3]{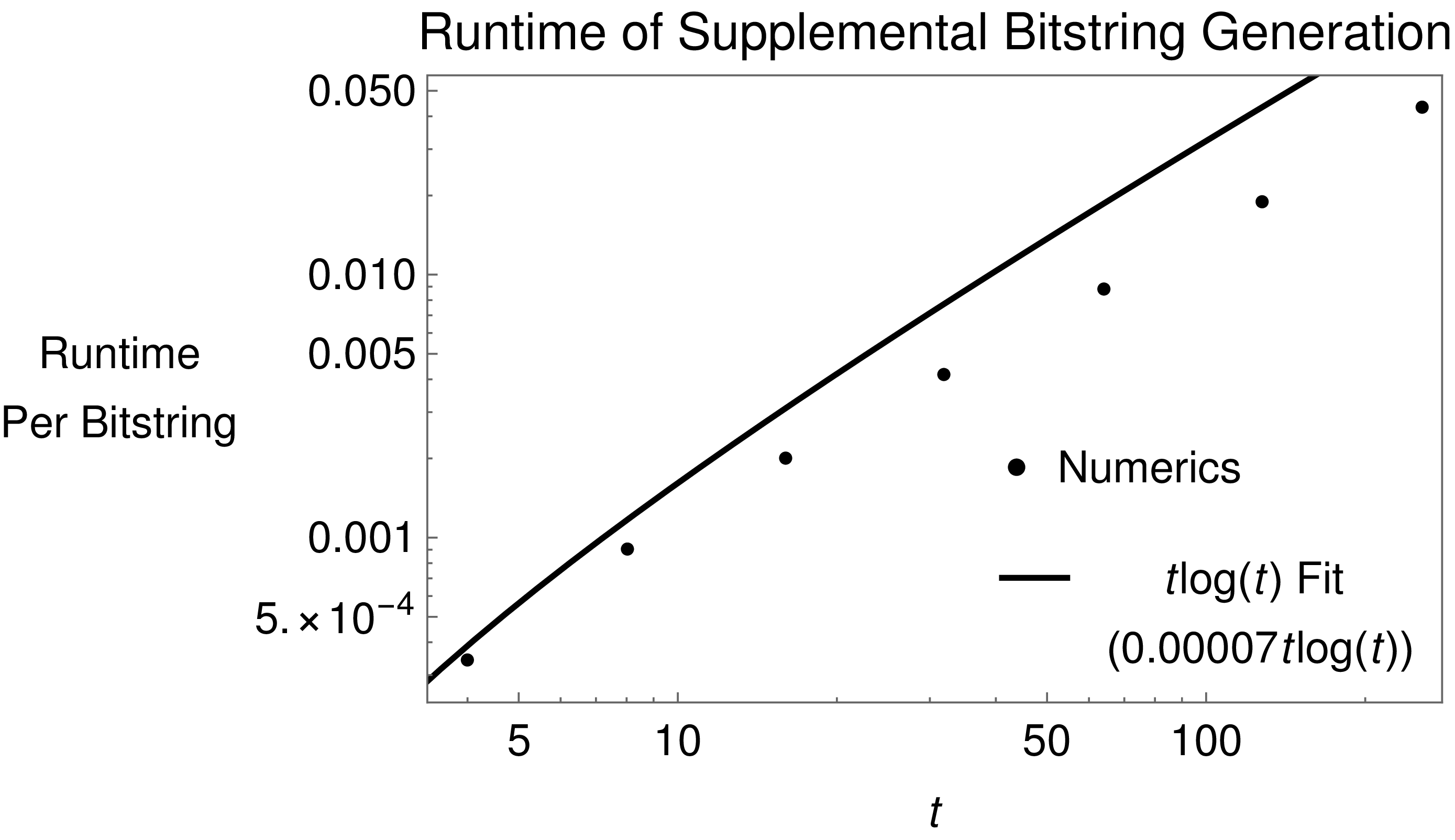}
\caption{Plot depicting \(t\log t\) runtime of Algorithm~\ref{alg:2tbitflips} implemented by the R code in Appendix~\ref{app:linear_scaling_code}.}
\label{fig:2tgenerationtiming}
\end{figure}

More than \(2t-1\) \(t\)-bitstrings can be effectively generated by using longer bitstrings and tracing out the extra bits at the end of the calculation, but this comes with cubic asymptotic overhead (see Appendix~\ref{app:superlinear_scaling} for Lemma~\ref{le:general_generation}). Figure~\ref{fig:supplemental_generation} shows that this allows for supplementing \(t \lesssim 24\) one-qubit magic states for \(\delta = 0.1\) with no overhead, and a larger range of \(t\) for smaller values of \(\delta\). These results can be trivially extended to generate supplemental states for all even \(t\) (see Appendix~\ref{app:superlinear_scaling} for Corollary~\ref{co:general_generation}).

\begin{figure}[t]
  \includegraphics[scale=0.6]{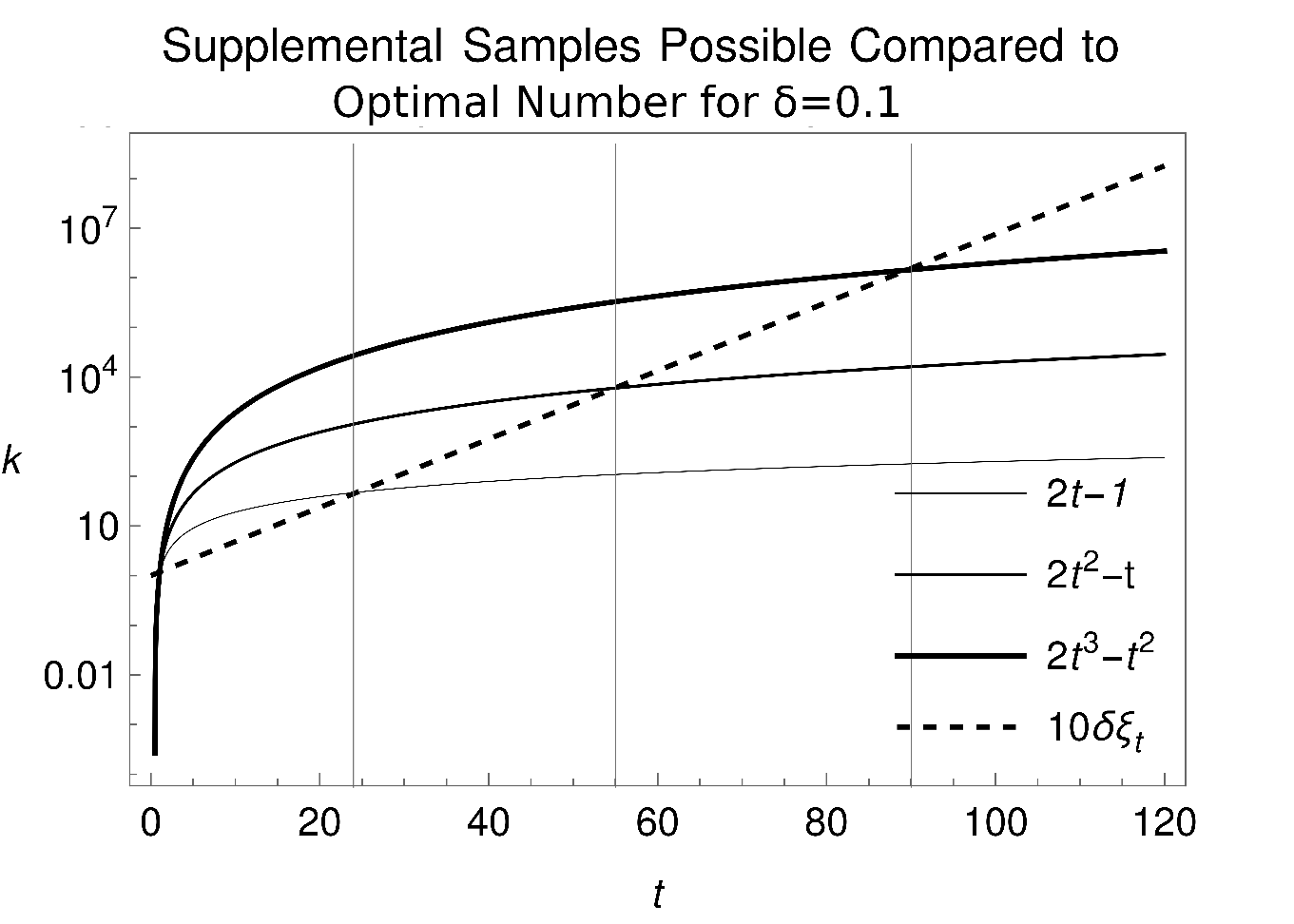}
\caption{The number of correlated supplemental states subject to the constraint \(\mathbb E(\braket{\tilde x}{\tilde y(x;j)}) \in o(\xi_t(\phi)^{-1})\) that can be produced is \(2t-1\) (from Lemma~\ref{le:linear_generation}). This can be further increased by \(\alpha\) according to Lemma~\ref{le:general_generation} at the expense of a prefactor of \(\alpha^3\). Supplemental state scaling of \(2t-1\), \((2t-1)t\) and \((2t-1)t^2\) are plotted as solid curves. The number required for optimal weak simulation of \(\rho_1\) and the \(T\) gate matic state is \(10 \delta \xi_t\) (dashed line is for \(\delta = 0.1\)). Therefore, an adequate number of supplemental correlated states can be produced by Lemma~\ref{le:linear_generation} methods for \(\delta=0.1\) when \(t\lesssim 22\), and higher \(t\) values must use Lemma~\ref{le:general_generation}. For \(\delta<0.1\), the required number (dashed curve) is reduced and so the supplementable domain of Lemma~\ref{le:linear_generation} of \(t\) increases.}
  \label{fig:supplemental_generation}
\end{figure}

\section{Correlated Weak Simulation}
\label{correlated_weak_sim}

Though supplementing with correlated stabilizer states that satisfy \(\mathbb E(\braket{\tilde x}{\tilde y(x;j)}) \in o(\xi_t(\phi)^{-1})\) reduces the expectation value given in Eq.~\ref{eq:expval}, it also increases its variance. If the distribution is already sufficiently peaked because \(\delta\) is small enough, this increase is negligible but as we have already seen in Section~\ref{sparsification}, this still only leads to a next-order in \(t\) reduction, \(\mathcal O((\xi_t - f_t)\delta^{-2})\), which is asymptotically negligible. This is formalized in the following theorem:
  \begin{theorem}[Correlated Sampling of \(\ket{\Psi}\)]
  \label{th:sparsifyupperbound}
  Given sets of \(f_t+1\) correlated stabilizer states \(\{\omega_i\}_{i=1}^{f_t}\) that satisfy \(\mathbb E(\braket{\omega_i}{\omega_j}) < \xi_t^{-1}\) and decompose the \(t\)-qubit state \(\ket{\Psi}\) with coefficients \(\|c\|_1^2\) saturating its stabilizer extent \(\xi_t\), a classical algorithm exists that creates a \(\delta\)-approximate stabilizer decomposition \(\ket{\psi}\) of \(\ket{\Psi}\),
  \begin{equation}
    \label{eq:th1}
    \|\ket{\Psi} - \ket{\psi}\| \le \delta,
  \end{equation}
  with \(\mathcal O((\xi_t{-}f_t) \delta^{-2} )\) states for \(t\) sufficiently large such that \(\delta^2 \gg (\xi_t -f_t)^{-1}\).
\end{theorem}
See Appendix~\ref{app:nextorderintsamplingbounds} for more details of this derivation and Appendix~\ref{app:numericsofnextorderint} for numerical results that support the proofs.

To increase the effect from correlated sampling, we might consider renormalizing the sampling ensemble. Heuristically, this will upgrade next-order effects to leading-order. In fact, renormalization produces the current state-of-the-art. By using the renormalized density matrix, \(\rho_1 \equiv \mathbb E \left(\frac{\ketbra{\psi}{\psi}}{\braket{\psi}{\psi}}\right)\), the number of i.i.d.~samples required to obtain \({\le} \delta\) error decreases from \(\mathcal O(\xi_t \delta^{-2})\) to \((2+\sqrt{2})\xi_t\delta^{-1}\) for Clifford magic states~\cite{Seddon20}. The transition to density matrices changes the equivalent error threshold from \(\delta^2\) to \(\delta\) in Eq.~\ref{eq:th1}~\cite{Seddon20}.

This cost reduction by \(\delta^{-1}\) compared to prior methods is due to the renormalization averaging out coherent errors in each i.i.d.~ensemble~\cite{Seddon20}. Correlated samples should enhance this effect by producing coherent errors that, instead of averaging out, constructively add to decrease the total error.
We find this to be the case and that renormalized correlated sampling produces a leading-order reduction in \(t\): 
\begin{theorem}
  \label{th:ensemblebound}
  Given \(f_t\) correlated states that mutually satisfy \(\mathbb E(\braket{\tilde x_i}{\tilde y_{\sigma_i(j)}}) < \xi_t^{-1}\), then \(\lesssim 0.05 \xi_t \delta^{-1}\) states are sufficient in the sparsification of \(\ket{\psi}\) in order for \(\|\mathbb E\left[\frac{\ketbra{\psi}{\psi}}{\braket{\psi}{\psi}}\right] - \ketbra{\Psi}{\Psi}\| \le \delta\).
\end{theorem}
The proof is given in Appendix~\ref{app:ensemblesamplingbounds}. We note that this result applies for all states \(\ket{\Psi}\). However, when taken with Lemmas~\ref{le:linear_generation} and~\ref{le:general_generation} and Corollary~\ref{co:general_generation}, it is a \emph{constructive} proof only for one-qubit magic states \(\ket{\Psi_\phi}\) and mixed states of one-qubit magic states.

As shown in its proof, Theorem~\ref{th:ensemblebound} produces a reduction by a factor of \({\sim} 68\) over the current state-of-the-art upper bound~\cite{Seddon19}. This is equivalent to removing \({\sim}27\) magic gates for the \(T\) gate (\(\phi = \pi/4\)),
\begin{equation}
68 = \xi^t(\pi/4) = 2^{\sim 0.228 t} \Leftrightarrow t \approx 27. 
\end{equation}
A numerical demonstration of reduced scaling are provided in Figure~\ref{fig:numerics2}.

\section{Performance Analysis}
\label{performance}

\begin{figure}[h]
\includegraphics[scale=0.35]{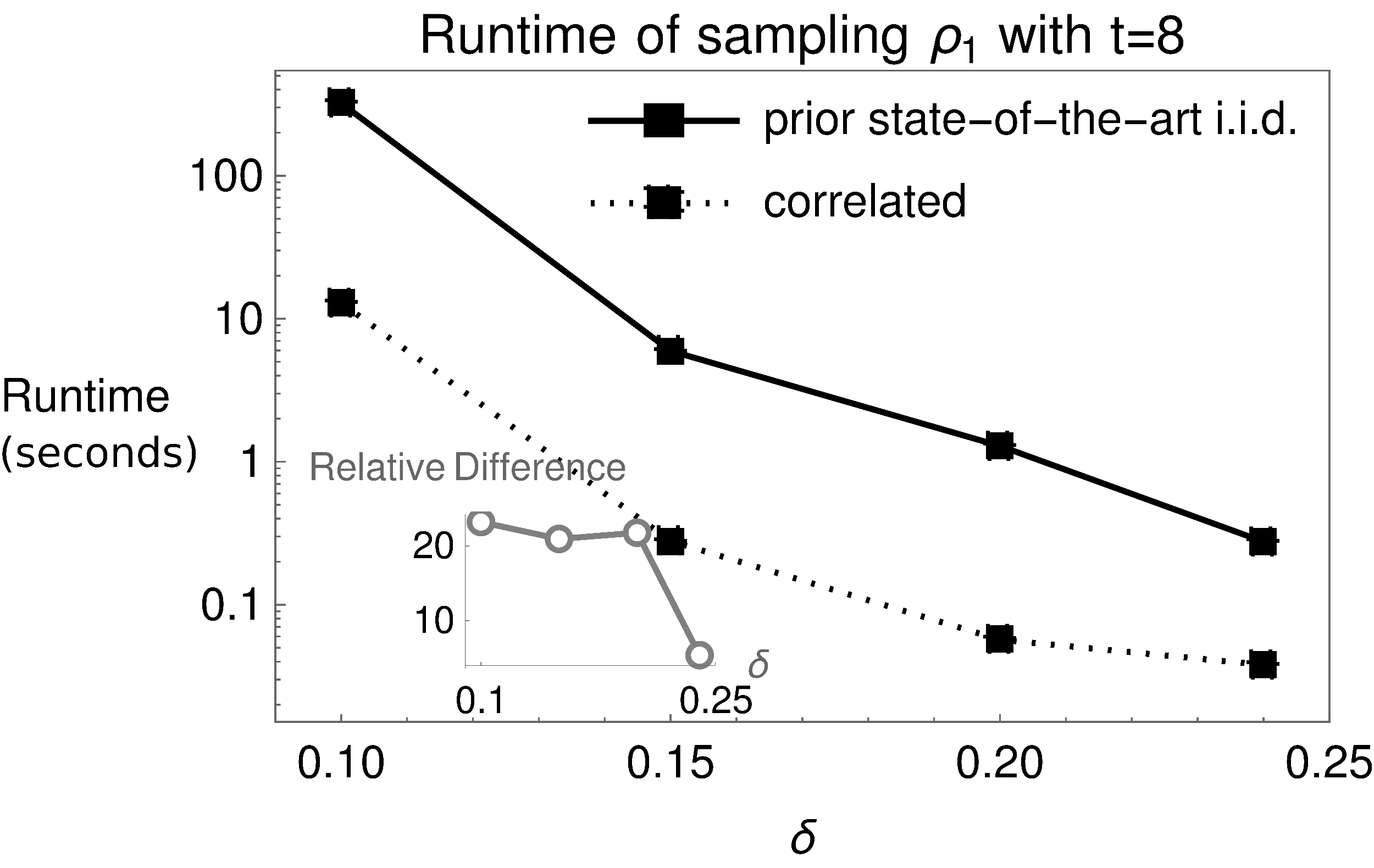}
\caption{Worst case runtime of prior state-of-the-art error of estimated probabilities after \(1000\) Cliffords and two marginal measurements vs \(\delta\) for \(t=8\).}
\label{fig:numerics2}
\end{figure}

Converged numerical studies of worst-case scaling of methods that include i.i.d.~sampling (which the correlated runs also use) are difficult to probe; the number of possible i.i.d. sampling combinations grows exponentially with both \(t\) and \(\delta^{-1}\). These results are not indicative of average-case performance.

As a result, Figure~\ref{fig:numerics2} presents worst-cast performance numerics for a low value of \(t\) and low values of \(\delta^{-1}\) only, which is well before asymptotic scaling behavior has set it. Nevertheless, Figure~\ref{fig:numerics2} shows significant scaling reduction of correlated sampling (Lemma~\ref{th:ensemblebound_general}) compared to the prior state-of-the-art. As can be seen in inset of Figure~\ref{fig:numerics2}, in this limited interval of \(\delta \in [0.1,0.24]\), correlated sampling worst-case performance is more efficent than the prior state-of-the-art worst-case performance by a factor of \({\sim}32\) at the smallest value \(\delta = 0.1\).

The correlated sampling algorithm described by Theorem~\ref{th:ensemblebound} is so significant that it even outperforms current state-of-the-art strong simulators. Since probabilites range from \(0\) to \(1\), additive error \(\delta\) must always be less than or equal to multiplicative error, \(\epsilon\), to express the same value. The lowest worst-case cost of weak simulators occurs when this inequality is saturated and \(\epsilon = \delta\). As indicated earlier, if \(\delta\) becomes too small, it should become more efficient to use strong simulators instead of weak simulators. At the other end, if \(\epsilon\) gets too small, then it should become more efficient to use exact simulators instead of strong simulators.

\begin{figure}[h]
  \includegraphics[scale=0.35]{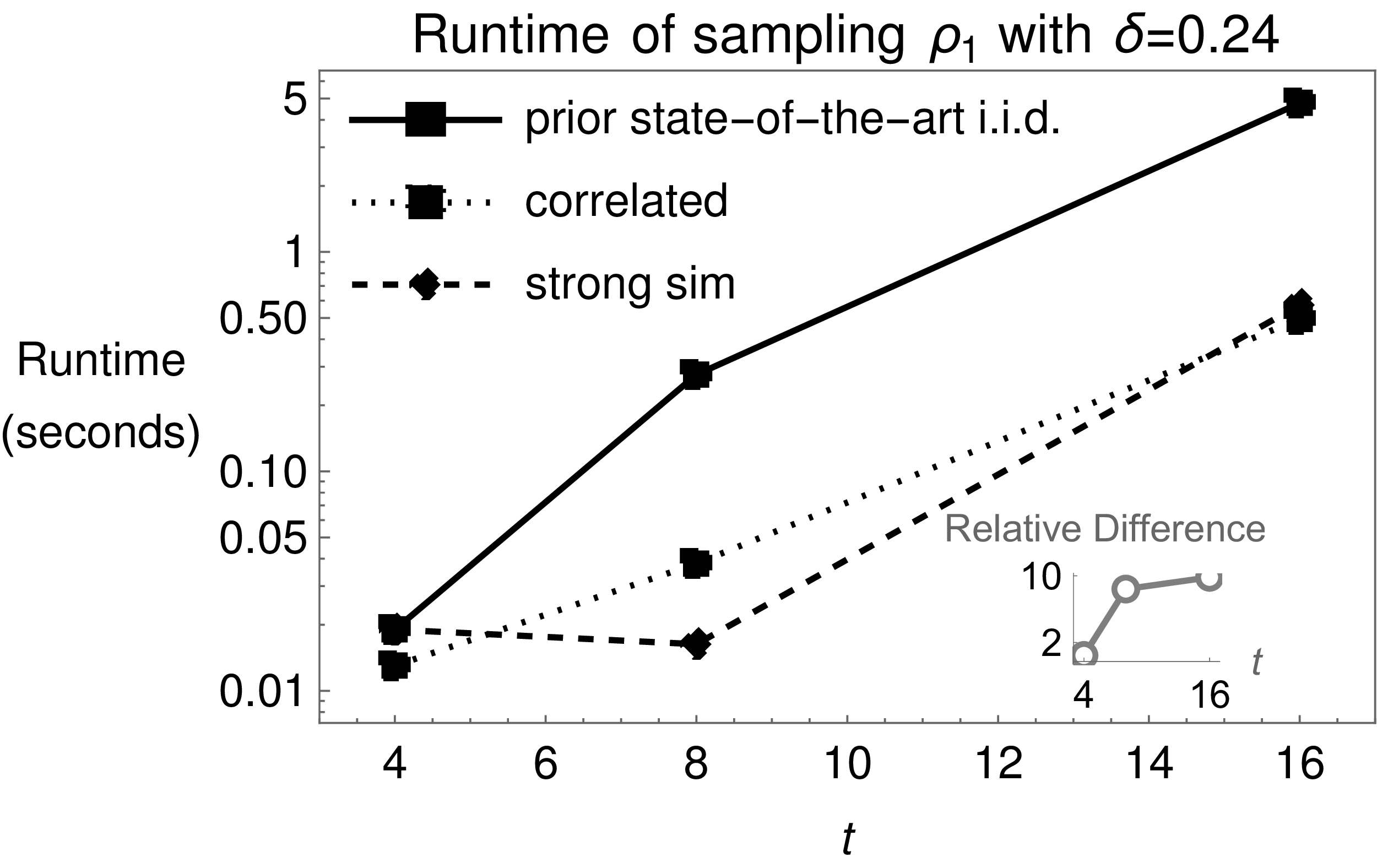}
  \includegraphics[scale=0.35]{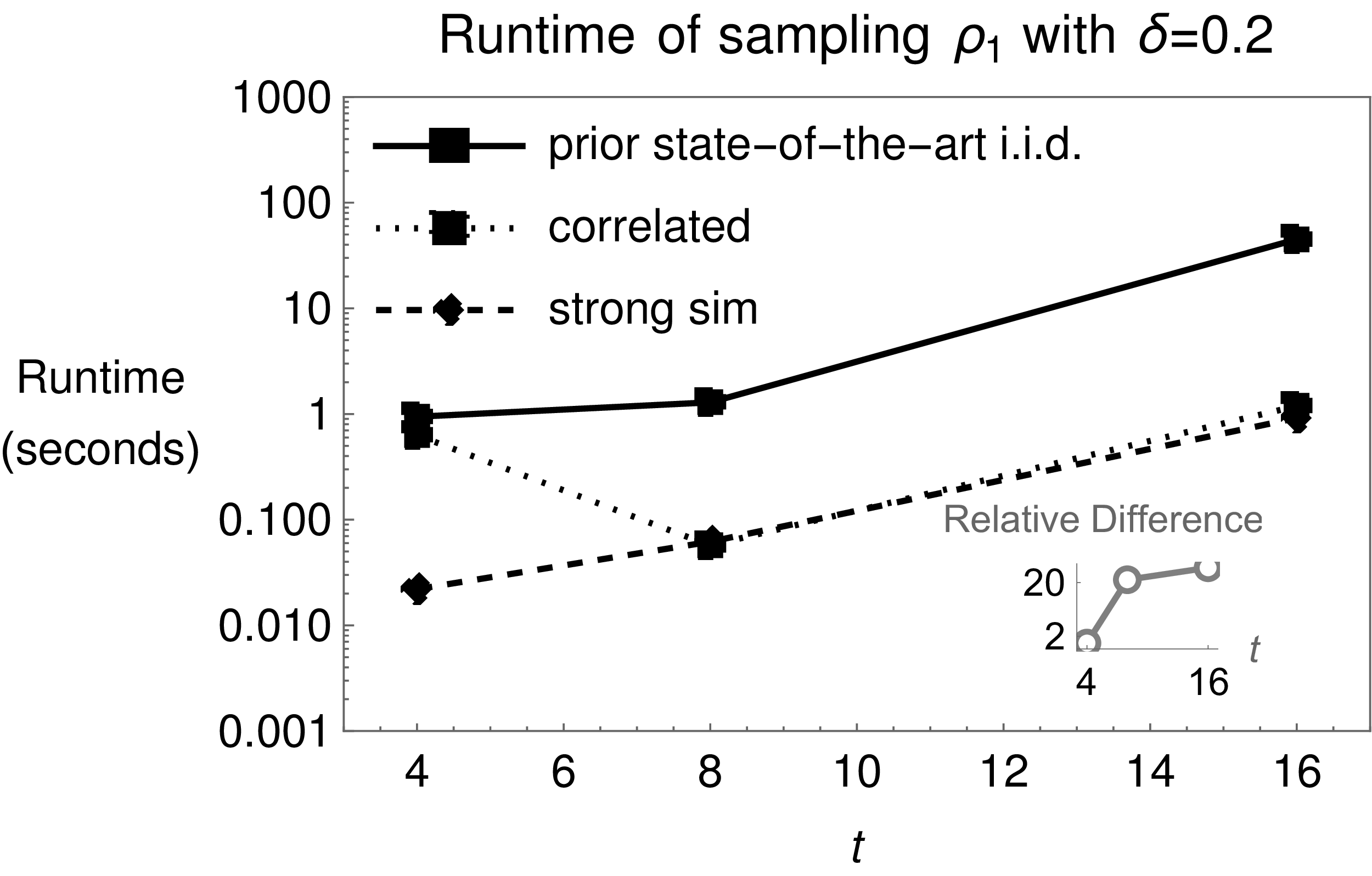}
  \includegraphics[scale=0.35]{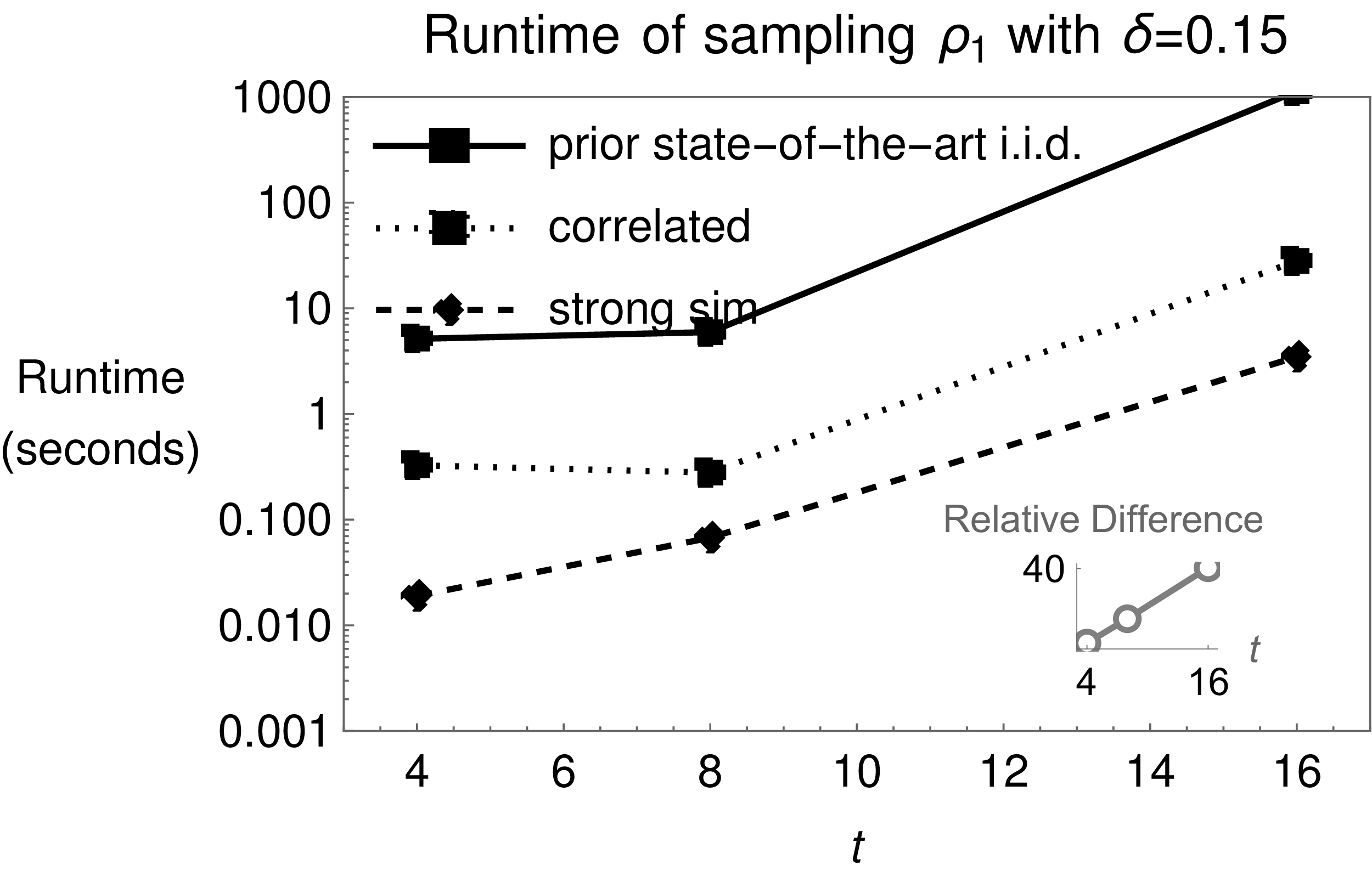}
\caption{Worst case runtime of prior state-of-the-art error of estimated probabilities after \(1000\) Cliffords and two marginal measurements vs \(t\) for \(\delta=0.24\) (top), \(\delta=0.2\) (middle), and \(\delta = 0.15\) (bottom). The results shown are for \(\> 10000\) random two \(t\)-Pauli measurements probability estimations calculated as two marginals for \(t=\{4,8,16\}\). For the strong simulation runs, we used the lowest known stabilizer rank of \(\chi_4 = 4\), \(\chi_8 = 12\), and \(\chi_{16} = 108\)~\cite{Qassim21}.}
\label{fig:numerics}
\end{figure}

\begin{figure}[h]
  \includegraphics[scale=0.35]{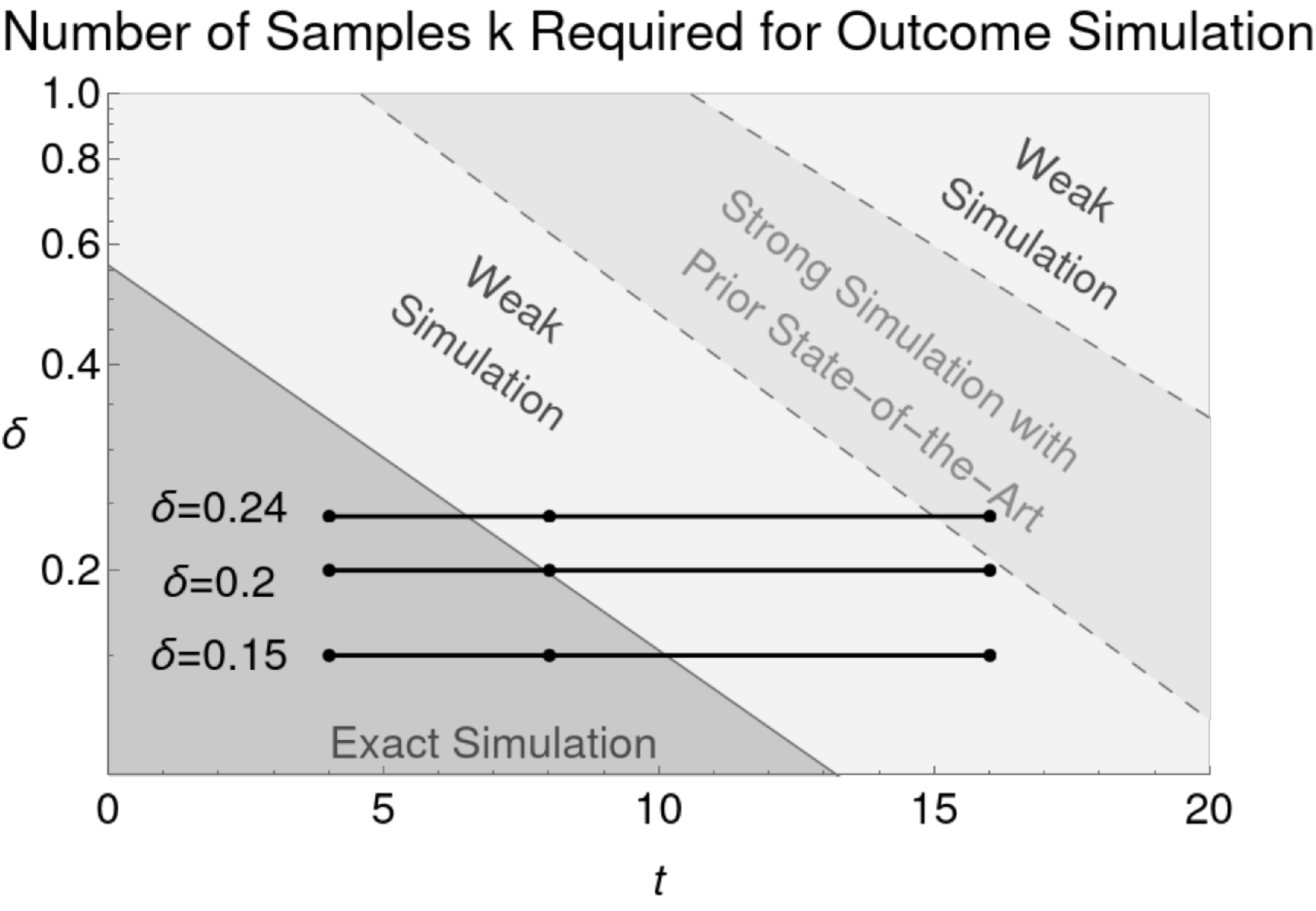}
\caption{Cross-section of Figure~\ref{fig:scaling_improvement}b explored by the runs plotted in Figure~\ref{fig:numerics}.}
\label{fig:crosssection}
\end{figure}

Applying this property to the current state-of-the-art strong simulation algorithm scaling~\cite{Bravyi16_1,Seddon20} means that the highest threshold for when strong and exact simulation can be more efficient than weak simulators occurs when
\begin{align}
  \label{eq:strongsimfewersamples}
  \chi_t <& (0.05)\xi_t \delta^{-1} \quad \substack{\text{strong simulation can use fewer states}\\ \text{than weak simulation}},\\
  \label{eq:exactsimfewersamples}
  \chi_t^2 <& (0.05)\xi_t \delta^{-3} \quad \substack{\text{exact simulation can use fewer states}\\ \text{than weak simulation}},
\end{align}
and more usefully,
\begin{align}
  \label{eq:strongsimlowercost}
  \frac{\chi_t/\delta^2}{2+\sqrt{2}} < (0.05&)\xi_t \delta^{-3} \, \substack{\text{strong simulation of outcomes}\\ \text{can cost less than weak simulation}},\\
  \label{eq:exactsimlowercost}
  \chi_t^2 < \frac{12 (0.05)}{2+\sqrt{2}}&\xi_t \delta^{-3} \, \substack{\text{exact simulation of outcomes}\\\text{can cost less than weak simulation}},
\end{align}
where \(\chi_t\) is the minimal exact stabilizer decomposition of the magic state~\cite{Bravyi16_1}. 
Remarkably, with the lowest current asymptotic upper bound on \(\chi_t = 2^{\sim 0.396 t}\) for \(T\) gates~\cite{Qassim21}, Eq.~\ref{eq:exactsimfewersamples} and Eq.~\ref{eq:exactsimlowercost} are satisfied at a higher value of \(\delta\) than Eq.~\ref{eq:strongsimfewersamples} and Eq.~\ref{eq:strongsimlowercost}, respectively, for \(t \lesssim 150\). Strong simulation with \(\chi_t = 2^{\sim 0.396 t}\) for \(t>150\) requires \(> 10^{17}\) stabilizer states and so is not practical. This means that the correlated sampling of \(\rho_1\) formally outperforms current strong simulators in any practically implementable regime.  This is not the case if the prior state-of-the-art weak simulation prefactor is used.

This comparison can be seen in Figure~\ref{fig:scaling_improvement} shown earlier, which plots Eqs.~\ref{eq:strongsimfewersamples}-\ref{eq:exactsimlowercost} w.r.t.~\(t\) and \(\delta\) with the current state-of-the-art for the \(T\) gate. Moreover, we find that the correlated sampling of \(\rho_1\) outperforms even the average-case runtime of prior state-of-the-art \(L_1\)-norm based simulators (see Appendix~\ref{app:otherstateoftheart}).

Applying Eqs.~\ref{eq:strongsimfewersamples}-\ref{eq:exactsimlowercost} at small \(t\) and \(\delta^{-1}\) assumes sufficiently similar or faster onset of asymptotic behavior for the correlated sampling of \(\rho_1\) compared to the state-of-the-art weak and strong simulation algorithms. Numerical results at values of \(t\) and \(\delta\) that cross through the domains when weak simulation, strong simulation, or exact simulation should be optimal support this assumption as shown in Figure~\ref{fig:numerics} for additive error \(\delta = 0.24\) (top), \(\delta = 0.2\) (middle) and \(\delta = 0.15\) (bottom) for \(t \in \{4,8,16\}\). These values of \(\delta\) and \(t\) trace cross-sections through a patch of Figure~\ref{fig:scaling_improvement} that is shown in the close-up version given by Figure~\ref{fig:crosssection}. For this last value \(\delta = 0.15\), correlated sampling is more efficient at \(t=16\) by a factor \(\gtrsim 40\).

This performance allows correlated weak sampling to outperform the state-of-the-art strong simulation algorithm for \(t \lesssim 150\) as predicted by its asymptotic scaling. This is clear from the dashed ``strong sim'' curve plotted in Figure~\ref{fig:numerics}; the state-of-the-art strong simulation algorithm is less efficient than the new correlated sampling weak simulation algorithm for \(\delta = 0.24\) and \(\delta = 0.2\), but becomes more efficient at \(\delta = 0.15\), for the values \(t=\{4, 8, 16\}\). This means that Figure~\ref{fig:scaling_improvement} is essentially indiscernible from a more accurate plot that does not use asymptotic scaling at low values of \(t\) and \(\delta^{-1}\).


Moreover, Figures~\ref{fig:numerics2}-\ref{fig:crosssection} show that the difference in worst-case cost between the correlated weak simulation algorithm and the prior state-of-the-art does not approach asymptotic behavior until \(t\gtrsim 8\) and \(\delta \lesssim 0.1\).

As an aside, note that probability estimation using stabilizer decompositions generated by strong or weak simulators requires calculating inner products, which are most efficiently performed in the worst-case by algorithms that scale with respect to the overall corresponding multipicative error desired~\cite{Bravyi16_1}. 
We did not use such algorithms to calculate the inner products~\cite{Bravyi16_1, Seddon20} and instead calculated them exactly, at the additional expense of squaring the overall scaling cost of every method (i.e. \((2+\sqrt{2})(\xi_t \delta^{-1}) \rightarrow (2+\sqrt{2})^2\xi^2_t \delta^{-2}\)). Using this less efficient method does not affect comparisons of the weak and strong algorithms' relative performance, which is our primary goal. However, this means that the absolute performances of the algorithms shown in Figures~\ref{fig:numerics2}-\ref{fig:crosssection} should not be taken as representative of their respective worst-case performance when implemented as efficiently as possible. 

\section{Conclusion}
\label{conclusion}

This result points to the need to further lower known optimal stabilizer decompositions (i.e. stabilizer rank) of magic states~\cite{Bravyi16_2,Howard18,Kocia20_2,Qassim21} to render direct strong simulation practically useful.

Further work can be done to extend these results so that Theorem~\ref{th:ensemblebound} is also constructive for magic states consisting of more than one-qubit tensor products. Furthermore, it may be possible to tighten some of the bounds made during the derivation of Theorem~\ref{th:ensemblebound} with more accurate assumptions, and thereby produce a lower overall asymptotic bound. Further work can also be done to increase the reduction from correlated sampling of \(\rho_1\) at non-asymptotic values of \(t\). Regardless, this method of correlated sampling greatly increases the size of quantum devices that can be directly simulated classically. An appealing feature of the correlated algorithm is that it can be easily added to the stabilizer state decompositions used by many contemporary weak simulators~\cite{Bravyi16_1,Howard18,Seddon20,Pashayan21} and reduce their cost, since it only involves supplementing their decompositions.

In summary, we prove that classical simulation of quantum outcomes to additive error \(\delta\) have a leading-order reduction in their sampling cost from \(\le (2+\sqrt{2})\xi_t \delta^{-1}\) to \(\lesssim 0.05 \xi_t \delta^{-1}\), a reduction by a factor of \({\sim}68\). This outperforms current state-of-the-art weak and strong simulators. This reduction is based on supplementing independent \(L_1\) sampling with correlated sampling of maximally mutually ``dissimilar'' states. \(L_1\)-norm based simulation techniques are also widely used in signal processing~\cite{Bobin08,Eldar12}, chemistry~\cite{Sanders12,Gamez16,Wang20}, and condensed matter physics~\cite{Liu21_2}. Remarkably, it appears that this correlated technique has not been applied in these fields and will likely lower to cost of applications such as time propagation and spectroscopy using Gaussian sampling~\cite{Wang01,Miller02}, gate compilation from ensembles~\cite{Patel21}, and randomized benchmarking of quantum devices~\cite{Magesan11,Proctor17}. Since it only involves supplementing existing i.i.d.~samples, the method can be easily added to existing simulators.

\bibliography{biblio}{}
\bibliographystyle{unsrt}

This research was supported in part by an appointment to the National Nuclear Security Administration (NNSA) Minority Serving Institutions Internship Program (MSIIP) sponsored by the NNSA and administered by the Oak Ridge Institute for Science and Education (ORISE).

This material is based upon work supported by the U.S. Department of Energy, Office of Science, Office of Advanced Scientific Computing Research, under the Accelerated Research in Quantum Computing program. Sandia National Laboratories is a multimission laboratory managed and operated by National Technology \& Engineering Solutions of Sandia, LLC, a wholly owned subsidiary of Honeywell International Inc., for the U.S. Department of Energy’s National Nuclear Security Administration under contract DE-NA0003525. This paper describes objective technical results and analysis. Any subjective views or opinions that might be expressed in the paper do not necessarily represent the views of the U.S. Department of Energy or the United States Government.

\acknowledgments{L.K. and G.T. gratefully acknowledge Mohan Sarovar's helpful correspondence throughout this project and L.K thanks Andrew Baczewski for help with the manuscript.}

\section*{Publishing Data Sharing Policy}

The code that support the findings of this study is openly available at `https://s3miclassical.com/gitweb/', in the repository `weak\_simulation\_extent.git'.

\clearpage

\appendix

\section{Linear Scaling of \(f_t\)}
\label{app:linear_scaling}

\setcounter{lemma}{0}
\begin{lemma}[\(2t-1\) Generation]
Given a \(t\)-bitstring, there exist \(2t-1\) additional bitstrings that mutually differ from each other and the given bitstring by at least \(t/2\) bitflips and can be found in subquadratic time.
\end{lemma}
\begin{proof}
Let $\alpha = \{10\}$, $\beta = \{01\}$, $\gamma = \{00\}$, $\delta = \{11\}$.
WLOG we consider $x$ as the bitstring of all $1$'s since we can always consider supplementing any other given bitstring by XORing the complements of the generated additional bitstrings.

We will prove this by induction for \(t\) powers of \(2\). We will show the \(t=2\) case is true, assume that the \(t=2^k\) is true, and then prove that \(t=2^{k+1}\) case is true. We refer to additional bitstrings as \(y\).

When $t=2$ the $y$'s are $\alpha$, $\beta$ and $\gamma$. These satisfy the previous inequalities.
  
We assume true for $t=2^k$, for some $k>1$. In other words, for a given bitstring \(x\) of length $2^k$, we assume we generate $(2(2^k)-1)$ $y_i$'s $\rightarrow (2^{k+1}-1)$ $y_i$'s. 

We now consider $t=2^{k+1}$.
From the prior assumption step, we are given $x\cdot y_i\leq \frac{1}{\sqrt{2}^{2^k}} = \frac{1}{\sqrt{2}^{t/2}}$ and $y_i\cdot y_j\leq \frac{1}{\sqrt{2}^{2^{k+1}}} = \frac{1}{\sqrt{2}^t}$ for $i\ne j$. 
For a bitstring $x x$ with of length $2^{k+1}$, we want $(2(2(2^k))-1)$ $y_i$'s $= 2(2^{k+1})-1$ $y_i$' s.

We define the operation of cloning. Cloning $x$ allows us to get $xx$ where $x\rightarrow xx$ taking our bitstring length $2^k \rightarrow 2(2^k) = 2^{k+1}$ and similarly, cloning $y_i$ allow us to get $y_iy_i$ where $y_i\rightarrow y_i y_i$ taking our bitstring length $2^k\rightarrow 2(2^k)=2^{k+1}$.

After cloning we get
\begin{equation}
  xx\cdot y_iy_i \leq \frac{1}{\sqrt{2}^{2^{k}}} = \frac{1}{\sqrt{2}^{t/2}},
\end{equation}
and
\begin{equation}
  y_iy_i\cdot y_jy_j \leq \frac{1}{\sqrt{2}^{2^{k+1}}} = \frac{1}{\sqrt{2}^t}, i \ne j.
\end{equation}

We define the operation of complementation, which takes $y_i\rightarrow \overline{y_i} = y_j$, where $j\ne i$:
for \(t=2\) bitstrings, $\overline{\alpha}= \beta$, $\overline{\beta} =\alpha$, $\overline{\gamma}= \delta$, and $\overline{\delta} = \gamma$, and higher \(t=2^k\) transform correspondingly. This operation flips all the bits.

We now consider the $t$ additional bitstrings produced from all cloned bitstrings, which additionally have none or one of their components complemented. From our assumption, there are $2(2^k)-1$ bitstrings produced with no complemented components. Complementation of the lower bits produces $2(2^k)-1$ additional bitstrings. This brings us to a total of $2 (2 (2^k) - 1) = 2 (2^{k+1}) - 2$. The final additional bitstring we consider is of all $\delta$'s with the lower bits complemented, bringing the total to $2 (2^{k+1}) - 1 = 2 t - 1$ additional prospective bitstrings. We now proceed to show that these are all valid.

$xx\equiv x\cdot 2^{2^k} +x$.  Similarly,  $y_iy_i\equiv y_i\cdot 2^{2^k} +y_i. 
\\xx\cdot y_iy_i\rightarrow (x\cdot2^{2^k}+x)\cdot (y_i\cdot 2^{2^k}+y_i) = (x\cdot 2^{2^k})\cdot(y_i\cdot 2^{2^k})\times(x\cdot y_i)$. We know that $(x\cdot 2^{2^k})\cdot (y_i\cdot 2^{2^k})\leq \frac{1}{\sqrt{2}^{2^{k-1}}}$  and $(x\cdot y_i)\leq \frac{1}{\sqrt{2}^{2^{k-1}}}$. Therefore $\frac{1}{\sqrt{2}^{2^{k-1}}}\times \frac{1}{\sqrt{2}^{2^{k-1}}}  = \frac{1}{\sqrt{2}^{2^k}} = \frac{1}{\sqrt{2}^{t/2}}.$
$xx\equiv x\cdot 2^{2^k}+x.$ Similarly, $y_i y_j\equiv y_i\cdot 2^{2^k} + y_j$.
$xx\cdot y_iy_j\rightarrow (x\cdot 2^{2^k}+x)\cdot (y_i\cdot 2^{2^k}+y_j) = (x\cdot 2^{2^k})\cdot (y_i\cdot 2^{2^k})\times (x\cdot y_j).$
We know that  $(x\cdot 2^{2^k})\cdot (y_i\cdot 2^{2^k})\leq \frac{1}{\sqrt{2}^{2^{k-1}}}$  and $(x\cdot y_j)\leq \frac{1}{\sqrt{2}^{2^{k-1}}}.$ 
Therefore $\frac{1}{\sqrt{2}^{2^{k-1}}} \times \frac{1}{\sqrt{2}^{2^{k-1}}} = \frac{1}{\sqrt{2}^{2^k}} 
\\y_iy_i \equiv y_i\cdot 2^{2^k}+ y_i.$  Similarly, $ y_j y_j\equiv y_j\cdot 2^{2^k} + y_j.
\\y_iy_i\cdot y_jy_j \rightarrow (y_i\cdot 2^{2^k} + y_i)\cdot (y_j\cdot 2^{2^k} +y_j) = (y_i\cdot 2^{2^k})\cdot (y_j\cdot 2^{2^k}) \times (y_i\cdot y_j).$ 
We know that $(y_i\cdot 2^{2^k})\cdot (y_j\cdot 2^{2^k}) = \frac{1}{\sqrt{2}^{2^{k}}}$ and $(y_i\cdot y_j) \leq \frac{1}{\sqrt{2}^{2^{k}}}.$ 
Therefore $\frac{1}{\sqrt{2}^{2^{k}}} \times\frac{1}{\sqrt{2}^{2^{k}}} = \frac{1}{\sqrt{2}^{2^{k+1}}} =\frac{1}{\sqrt{2}^{t}}$.

From the same reasoning, it is clear that \(y_i y_i \cdot y_i y_j = 1 \times\frac{1}{\sqrt{2}^{2^{k}}} = \frac{1}{\sqrt{2}^{2^{k}}} =\frac{1}{\sqrt{2}^{t/2}}\).

These are all the combinations formed at \(t=2^{k+1}\) (see Table~\ref{tab:binarytree} for examples) and so this completes the proof.\qed
\begin{table}[H]
\begin{tabular}{c|c|c|c}
A.\\
\(t\) & given bitstring & depth \(i\) & \(t\) additional bitstrings\\
\hline
\(2\) & \(11\) & \(0\) & \(\alpha\), \(\beta\) \\
\hline
\multirow{2}{*}{\(4\)} & \multirow{2}{*}{\(1111\)} & \(0\) & \(\alpha \alpha\), \(\beta \beta\)\\
\cline{3-4} & & \(1\) & {\(\alpha \beta\), \(\beta \alpha\)}\\
\hline
\multirow{3}{*}{\(8\)} & \multirow{3}{*}{\(11111111\)} & \(0\) & \(\alpha \alpha \alpha \alpha\), \(\beta \beta \beta \beta\)\\
  \cline{3-4}
  & & \(1\) & \(\alpha \alpha \beta \beta\), \(\beta \beta \alpha \alpha\) \\
  \cline{3-4}
  & & \(2\) & \(\alpha \beta \beta \alpha\), \(\beta \alpha \alpha \beta\), \(\alpha \beta \alpha \beta\), \(\beta \alpha \beta \alpha\)\\
\hline
\multirow{8}{*}{\(16\)} & \multirow{8}{*}{\(1111111111111111\)} & \(0\) & \(\alpha \alpha \alpha \alpha \alpha \alpha \alpha \alpha\), \(\beta \beta \beta \beta \beta \beta \beta \beta\)\\
  \cline{3-4}
      & & \(1\) & \(\alpha \alpha \alpha \alpha \beta \beta \beta \beta\), \(\beta \beta \beta \beta \alpha \alpha \alpha \alpha\)\\
  \cline{3-4}
      & & \multirow{2}{*}{\(2\)} & \(\alpha \alpha \beta \beta \alpha \alpha \beta \beta\), \(\beta \beta \alpha \alpha \beta \beta \alpha \alpha\),\\
      & & & \(\alpha \alpha \beta \beta \beta \beta \alpha \alpha\), \(\beta \beta \alpha \alpha \alpha \alpha \beta \beta\)\\
  \cline{3-4}
      & & \multirow{4}{*}{\(3\)} & \(\alpha \beta \alpha \beta \alpha \beta \alpha \beta\), \(\alpha \beta \alpha \beta \beta \alpha \beta \alpha\),\\
      & & & \(\beta \alpha \beta \alpha \alpha \beta \alpha \beta\), \(\beta \alpha \beta \alpha \beta \alpha \beta \alpha\),\\
      & & & \(\alpha \beta \beta \alpha \alpha \beta \beta \alpha\), \(\beta \alpha \alpha \beta \beta \alpha \alpha \beta\),\\
      & & & \(\alpha \beta \beta \alpha \beta \alpha \alpha \beta\), \(\beta \alpha \alpha \beta \alpha \beta \beta \alpha\)\\
\hline\hline
B.\\
\(t\) & given bitstring & depth \(i\) & \(t\) additional bitstrings\\
\hline
\(2\) & \(11\) & \(0\) & \(\gamma\) \\
\hline
\multirow{2}{*}{\(4\)} & \multirow{2}{*}{\(1111\)} & \(0\) & \(\gamma \gamma\)\\
\cline{3-4} & & \(1\) & {\(\gamma \delta\), \(\delta \gamma\)}\\
\hline
\multirow{3}{*}{\(8\)} & \multirow{3}{*}{\(11111111\)} & \(0\) & \(\gamma \gamma \gamma \gamma\)\\
  \cline{3-4}
  & & \(1\) & \(\gamma \gamma \delta \delta\), \(\delta \delta \gamma \gamma\) \\
  \cline{3-4}
  & & \(2\) & \(\gamma \delta \delta \gamma\), \(\delta \gamma \gamma \delta\), \(\gamma \delta \gamma \delta\), \(\delta \gamma \delta \gamma\)\\
\hline
\multirow{8}{*}{\(16\)} & \multirow{8}{*}{\(1111111111111111\)} & \(0\) &\(\gamma \gamma \gamma \gamma \gamma \gamma \gamma \gamma\)\\
  \cline{3-4}
      & & \(1\) & \(\gamma \gamma \gamma \gamma \delta \delta \delta \delta\), \(\delta \delta \delta \delta \gamma \gamma \gamma \gamma\)\\
  \cline{3-4}
      & & \multirow{2}{*}{\(2\)} & \(\gamma \gamma \delta \delta \gamma \gamma \delta \delta\), \(\delta \delta \gamma \gamma \delta \delta \gamma \gamma\),\\
      & & & \(\gamma \gamma \delta \delta \delta \delta \gamma \gamma\), \(\delta \delta \gamma \gamma \gamma \gamma \delta \delta\)\\
  \cline{3-4}
      & & \multirow{4}{*}{\(3\)} & \(\gamma \delta \gamma \delta \gamma \delta \gamma \delta\), \(\gamma \delta \gamma \delta \delta \gamma \delta \gamma\),\\
      & & & \(\delta \gamma \delta \gamma \gamma \delta \gamma \delta\), \(\delta \gamma \delta \gamma \delta \gamma \delta \gamma\),\\
      & & & \(\gamma \delta \delta \gamma \gamma \delta \delta \gamma\), \(\delta \gamma \gamma \delta \delta \gamma \gamma \delta\),\\
      & & & \(\gamma \delta \delta \gamma \delta \gamma \gamma \delta\), \(\delta \gamma \gamma \delta \gamma \delta \delta \gamma\)\\
\end{tabular}
\label{tab:binarytree}
\caption{\(\alpha \equiv 10\), \(\beta \equiv 01\), \(\gamma \equiv 00\), and \(\delta \equiv 11\). The complements of the additional bitstrings can be used as XOR masks to generate the appropriate supplemental bitstrings for a given bitstring other than all \(1\)s.}
\end{table}
\end{proof}

\section{R Code for \(2t-1\) Supplemental State Generation}
\label{app:linear_scaling_code}

{\small
\begin{verbatim}
#!/usr/bin/env Rscript

args = commandArgs(trailingOnly=TRUE)
t = as.numeric(args[1])

alpha <- c(1,0)
beta <-c(0,1)
gamma <-c(0,0)
delta <-c(1,1)

cloning <- function(vec)
{
  return <- array(c(vec,vec))
}
comp <-function(vec)
{
  for(i in 1:length(vec))
  {
    vec[i]= (vec[i]+ 1) %% 2
  }
  return <- vec
}
comp2 <-function(vec)
{
  for(i in (length(vec)/2 +1):length(vec))
  {
    vec[i]= (vec[i]+ 1) %% 2
  }
  return <- vec
}

yis <-list(alpha, beta, gamma)

for(k in 2:(log(t)/log(2)))
{
  newyis <-list()
 
  for(i in 1:length(yis))
  {
    newyis[[i]]<-(cloning(yis[[i]]))
  }
  for(i in 1:length(yis))
  {
    newyis[[i+length(yis)]]<-(comp2(cloning(yis[[i]])))
  }
 
  newyis[[2*length(yis)+1]]<-comp(comp2(newyis[[3]]))
  yis<- newyis
}

print(newyis)
\end{verbatim}
}

\section{Super-Linear Scaling of \(f_t\)}
\label{app:superlinear_scaling}

\begin{lemma}
\label{le:general_generation}
Given a \(t\)-bitstring for \(t\) a power of \(2\), the \(2t-1\) supplemental bitstrings of Lemma~\ref{le:linear_generation} that mutually differ by at least \(t/2\) bitstrings, can be increased to \(\alpha(2t-1)\) by a classical algorithm with additive \(\alpha t \log (\alpha t)\) cost per additional bitstring, which is asymptotically negligible. However, any subsequent calculations with these supplemental bitsrings requires an added overall multiplicative \(\mathcal O(\alpha^3)\) cost in the worst-case.
\end{lemma}
\begin{proof}
  Given a \(t\)-bitstring, and a desired number \(\alpha(2t-1)\) of supplemental bitstrings, use Lemma~\ref{le:linear_generation} to generate the \(2t'-1\) supplemental \(t'\)-bitstrings, where
  \begin{equation}
    2t'-1 = \alpha(2t-1) \leftrightarrow t' = \alpha(t-1/2)+1/2.
  \end{equation}

  As per Lemma~\ref{le:linear_generation}, this will require \(t' \log t' \approx \alpha t \log (\alpha t)\) cost per bitstring.

  The resultant supplemental \(t'\)-bitstrings can be used the same as \(t\)-bistrings in any calculation, taking care that the \(t'-t\) extra bits are not acted upon and are simply traced out at the end of any calculation. This requires calculating inner products that scale as \(\mathcal O(t'^3)\) and so require overall \(\mathcal O(\alpha^3)\) relative additional cost. 
\qed
\end{proof}

We note that Lemma~\ref{le:general_generation} is constructive.

  \begin{corollary}
  \label{co:general_generation}
Given a \(t\)-bitstring, where \(k\) is the lowest non-zero bit of \(t\) written in binary, there exist \(t'=2^{k+1}-1\) additional bitstrings that mutually differ from each other and the given bitstring by at least \(t/2\) bitflips and can be found in subquadratic in \(k\) time. This can be increased to \(\alpha t'\) by a classical algorithm with additive \(\alpha t' \log (\alpha t')\) cost per additional bitstring, which is asymptotically negligible. However, any subsequent calculations with these supplemental bitsrings requires an added overall multiplicative \(\mathcal O(\alpha^3)\) cost in the worst-case.
\end{corollary}

\section{Next-Order in \(t\) Sampling Bounds}
  \label{app:nextorderintsamplingbounds}

\setcounter{theorem}{0}
  \begin{theorem}[Correlated Sampling of \(\ket{\Psi}\)]
  Given sets of \(f_t+1\) correlated stabilizer states \(\{\omega_i\}_{i=1}^{f_t}\) that satisfy \(\mathbb E(\braket{\omega_i}{\omega_j}) < \xi_t^{-1}\) and decompose the \(t\)-qubit state \(\ket{\Psi}\) with coefficients \(\|c\|_1^2\) saturating its stabilizer extent \(\xi_t\), a classical algorithm exists that creates a \(\delta\)-approximate stabilizer decomposition \(\ket{\psi}\) of \(\ket{\Psi}\),
  \begin{equation}
    \|\ket{\Psi} - \ket{\psi}\| \le \delta,
  \end{equation}
  with \(\mathcal O((\xi_t{-}f_t) \delta^{-2} )\) states for \(t\) sufficiently large such that \(\delta^2 \gg (\xi_t -f_t)^{-1}\).
\end{theorem}
\begin{proof}
Following~\cite{Howard18}, we define additive error
\begin{equation}
\|\ket{\Psi} - \ket{\psi}\| \le \delta,
\end{equation}
where \(\|\psi\| \equiv \sqrt{\braket{\psi}{\psi}}\).

\(\ket{\psi}\) is the sparsified \(k\)-term approximation to \(\ket{\Psi}\) given by
\begin{equation}
  \ket \psi = \frac{\|c\|_1}{k} \sum_{i=1}^k \ket{\omega_i},
\end{equation}
where each \(\ket{\omega_i}\) is independently chosen randomly so that it is a normalized stabilizer state \(\ket{\omega_i} = p_i \ket{\varphi_i}\) with probability \(p_i = |c_i|/\|c\|_1\). We define a random variable \(\ket \omega\) that is equal to \(\ket{\omega_i}\) with probability \(p_i\). Then
\begin{equation}
  \mathbb E (\ket{\omega}) = \ket {\psi}/\|c\|_1.
\end{equation}

By construction,
\begin{equation}
\mathbb E(\braket{\psi}{\Psi}) = \mathbb E(\braket{\Psi}{\psi}) = 1.
\end{equation}

The number of stabilizer states in the approximation is \(k\) and
\begin{eqnarray}
  \mathbb E (\| \ket{\Psi} - \ket{\psi}\|^2) &=& \mathbb E( |\braket{\Psi}{\Psi}| ) - \mathbb E( |\braket{\Psi}{\psi}| ) -\nonumber\\
  \label{eq:kdependence}
  && \mathbb E( |\braket{\psi}{\Psi}| ) + \mathbb E( |\braket{\psi}{\psi}| ) \nonumber\\
  &\le& \frac{\xi_t}{k} - \frac{\gamma}{k},
\end{eqnarray}
where we simplified
\begin{eqnarray}
  \mathbb E(\braket{\psi}{\psi} ) &=& \sum_i^k \frac{\|c\|_1^2}{k^2} \mathbb E(\braket{\omega_i}{\omega_i}) + \sum_{i\ne j}^k \frac{\|c\|_1^2}{k^2} \mathbb E(\braket{\omega_i}{\omega_j}) \nonumber\\
  &=& \frac{\|c\|_1^2}{k} \mathbb E(|\braket{\omega}{\omega}|) + \sum_{i\ne j}^k \frac{\|c\|_1^2}{k^2} \mathbb E(\braket{\omega_i}{\omega_j}) \nonumber\\
  &\le& \frac{\|c\|_1^2}{k} + 1 - \frac{\gamma}{k}.
\end{eqnarray}

Eq.~\ref{eq:kdependence} is less than or equal to \(\delta^2\) when \(k = (\xi_\phi^t - \gamma) \delta^{-2}\).

If \(\ket{\omega_i}\) are i.i.d. stabilizer states then \(\sum_{i\ne j}^k \|c\|_1^2 \mathbb E(\|\braket{\omega_i}{\omega_j}) = \sum_{i\ne j}^k | \mathbb E(\bra{\psi}) \mathbb E(\ket{\psi}) | = k (k-1)\) and so \(\gamma = 1\). As a result, since \(1 \ll \xi_t\) as \(t\) increases, it was neglected in previous characterizations~\cite{Bravyi16_1}. 

However, \(\gamma\) can become significant if \(\ket{\omega_i}\) are not i.i.d. and \(\sum_{i\ne j}^k \|c\|_1^2 \mathbb E(\|\braket{\omega_i}{\omega_j}) \ne \sum_{i\ne j}^k | \mathbb E(\bra{\psi}) \mathbb E(\ket{\psi}) |\).

Let us supplement every i.i.d. element of this set with \(f_t\) states whose average mutual inner product is less than or equal to \(\xi_t^{-1}\).

With these supplemental states, we construct \(f_t\) additional ensembles, where the \(i\)th ensemble consists of the \(f_t\) \(i\)th supplemented states. Since the ensemble consisting of i.i.d~sampled \(t\)-bit strings \(\ket{\omega_i}\) satisfies
\begin{equation}
\mathbb E(\braket{\psi}{\Psi}) = \mathbb E(\braket{\Psi}{\psi}) = 1,
\end{equation}
it follows that the \(f_t\) other ensembles also satisfy this property since they also consist of only independent samples. Therefore, the full ensemble produced by adding together these \((f_t+1)\) ensembles satisfies this property too.

However, taken together, these are no longer i.i.d. stabilizer states. In particular, for a given \(\bra{\omega_i}\), there exist at least \(f_t\) \(\ket{\omega_j}\) such that \(\mathbb E |\braket{\omega_i}{\omega_j}| \le \xi_t^{-1}\). Hence,
\begin{eqnarray}
\label{eq:offdiagnorm}
\sum_{i\ne j}^k \|c\|_1^2 \mathbb E(\braket{\omega_i}{\omega_j}) &=& \sum_{i}^k \sum_{\substack{j\\i\ne j}}^{k-1-f_t} \|c\|_1^2 \mathbb E(\bra{\omega_i})\mathbb E(\ket{\omega_j}) \nonumber\\
                                                                 && + \sum_{i}^k \sum_{j}^{f_t} \|c\|_1^2 \mathbb E(\braket{\omega_i}{\omega_{\sigma_i(j)}}) \\
                                                                 &\le& k (k-1-f_t) \nonumber\\
                                                                 && + k \sum_j^{f_t} \xi_t(\phi) \mathbb E(\braket{\omega_i}{\omega_{\sigma_i(j)}}) \\
                                                                 &=& k(k-\gamma),
\end{eqnarray}
where \(\sigma_i(j)\) maps to the indices that index states correlated with \(\omega_i\). \(\|c\|^2_1 = \xi_t\) and so \(\gamma = 1 + f_t - \sum_j^{f_t} \xi_t \mathbb E(\braket{\omega_i}{\omega_{\sigma_i(j)}}\). In the limit \(t \rightarrow \infty\), \(\gamma \rightarrow f_t\).

Therefore, given that at least \(k = (\xi_\phi^t - \gamma) \delta^{-2}\) stabilizer states are necessary to sample this state to \(\delta\) additive error, this procedure creates a \(\delta\)-approximate stabilizer decomposition of \(\ket{\Psi}\) with \(\mathcal O((\xi_t{-} f_t) \delta^{-2} )\) stabilizer states.

This more efficiently approximated state comes at the expense of its convergence probability, or sparsification tail bound. Following the same reasoning as in the proof of Lemma \(7\) of~\cite{Howard18},
\begin{eqnarray}
  &&\Pr \left[ \|T^{\otimes t} - \psi\|^2 \le \braket{\psi}{\psi} - 1 + \delta^2 \right] \nonumber\\
  \label{eq:tailbound}
  &&\ge 1 {-} 2 \exp \left( -\frac{\delta^2 \xi_t}{8 } {+} \frac{\gamma \delta^2}{8} \right) \\
  \label{eq:tailbound_final}
  &&= 1 {-} 2 e^{-\frac{\delta^2 \xi_t}{8 } {+} (1 + f_t - \sum_j^{f_t} \xi_t \mathbb E(\braket{\omega_i}{\omega_{\sigma_i(j)}})) \frac{\delta^2}{8}}.
\end{eqnarray}

Therefore, given that \(\delta^2 \gg (\xi_t -f_t)^{-1}\), if post-selection is performed to discard samples that produce \(\braket{\psi}{\psi} - 1 \gg \delta^2\) (a rare event if this first condition is met) or \(\braket{\psi}{\psi}\) is approximated to multiplicative error using the FASTNORM algorithm~\cite{Howard18} (which scales linearly with \(k\)), then the states \(\psi\) are generated with \(\mathbb E (\| \ket{\Psi} - \ket{\psi}\|^2) < \delta^2 \) and consist of \(\mathcal O((\xi_t {-} t) \delta^{-2})\) stabilizer states.
\qed
\end{proof}

\section{Numerics of Next-Order in \(t\) Reduction}
  \label{app:numericsofnextorderint}
  
  The numerical runtime of correlated sampling of \(\ket{T}^{\otimes t}\) is plotted in Figure~\ref{fig:normscaling}. A decrease in runtime is observed for correlated sampling that is lower bounded by proportionality to the fewer number of stabilizer states \(k = (\xi_t(\frac{\pi}{4}) - \gamma) \delta^{-2}\) it generates.

\begin{figure}[h]
\input{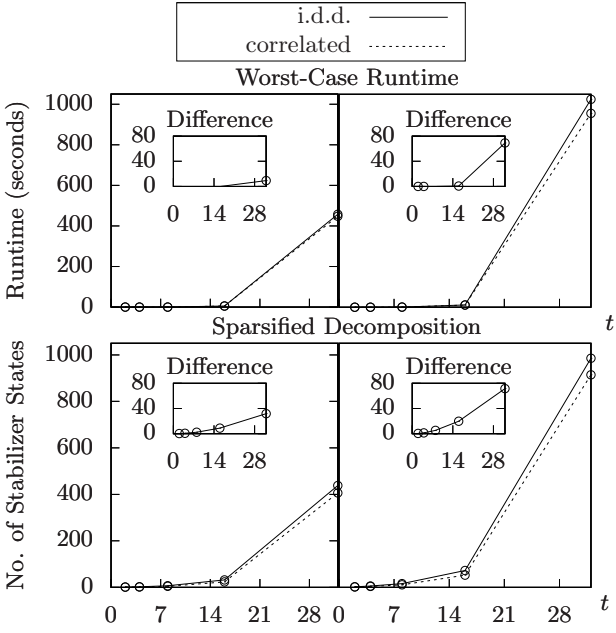}
\caption{Plots are for additive error (left) \(\delta = 0.6\) and (right) \(\delta = 0.4\) over \(100\) runs. The worst-case runtime of calculating the norm of \(\psi\) where \(|\psi {-} T^{\otimes t}| \le \delta^2\) is plotted at the top and the corresponding number of stabilizer states sampled in the sparsified decomposition with i.i.d. (solid curve) and correlated sampling (dashed curve) is plotted at the bottom. \(\psi\) is generated using the SPARSIFY algorithm and the norm is calculated using the FASTNORM algorithm of~\cite{Bravyi16_1} (using \(1000\) random stabilizer states to calculate the multiplicative error). [Insets: The difference between the i.i.d. sampling and the correlated sampling curves.]}
\label{fig:normscaling}
\end{figure}

The statistical distribution of sparsified decompositions from independent sampling and correlated sampling are compared over \(1000\) numerical runs in Figure~\ref{fig:meansandstdevs}. The expected value of the state generated by correlated sampling is closer to the desired state (middle of Figure~\ref{fig:meansandstdevs}). As a result, the standard deviation of the norm of the states generated by correlated sampling is larger (bottom of Figure~\ref{fig:meansandstdevs}) denoting poorer convergence, as expected. Hence, it is advantageous to use correlated sampling when \(\delta^2 \gg (\xi_t(\frac{\pi}{4}) -f_t)^{-1}\) to obtain the same convergence probability as independent sampling does at \(\delta^2 \gg \xi_{t}(\frac{\pi}{4})^{-1}\). 

At small \(t\) a small-number effect occurs since the number of stabilizer states used in correlated sampling is set to the nearest multiple of \(t\) greater than or equal to \((\xi_\phi^t-\gamma)\delta^{-2}\) in practice. As a result, at small \(t\) more states are sampled than required and this produces a lower expectation value and standard deviation than expected. The factor of \(\sum_j^{f_t} \xi_t(\frac{\pi}{4}) \mathbb E(\braket{\omega_i}{\omega_{\sigma_i(j)}})) \frac{\delta^2}{8}= -1/2^{\sim 0.02 t}\) in Eq.~\ref{eq:tailbound_final} also reduces the standard deviation at low \(t\).

\begin{figure}[h]
\input{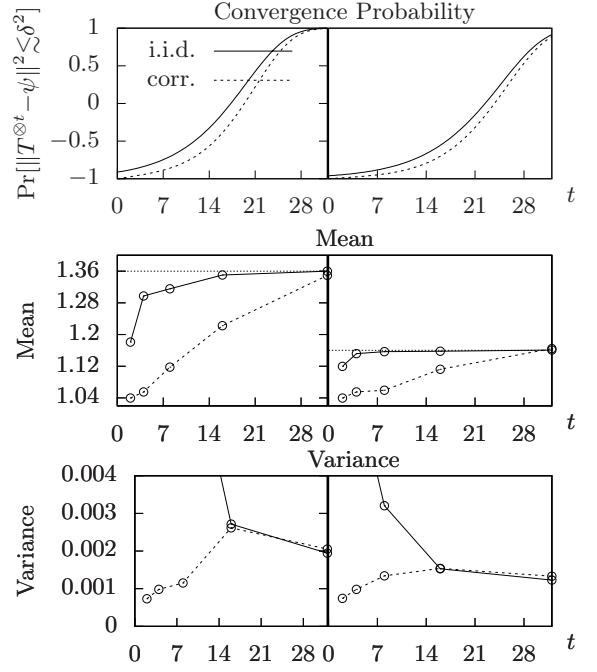}
\caption{Plots are for additive error (left) \(\delta = 0.6\) and (right) \(\delta = 0.4\) over \(1000\) runs. At the top is plotted the convergence probability lower bound (or sparsification tail bound) given by Eq.~\ref{eq:tailbound_final} for i.i.d. sampling (solid curve) and correlated sampling (dashed curve). In the middle is plotted the mean of \(|\braket{\psi}{\psi}|^2\) for both i.i.d. and correlated sampling to \(\delta\) error (\(1+\delta^2\) is denoted by the dotted horizontal line). At the bottom is plotted the standard deviation of \(|\braket{\psi}{\psi}|^2\) of the sparsified samples, which can be interpreted as a measure of the convergence probability. The standard deviation converges more slowly for large \(t\) under correlated sampling than under independent sampling. However, if the variance is sufficiengtly small already, then this effect is negligible. This requires that \(\delta^2 \gg \xi^{-t}\) and \(\delta^2 \gg(\xi^t - t)^{-1}\) for i.i.d. and correlated sampling. If satisfied, they exhibit \(\mathcal O(\xi^t \delta^{-2})\) and \(\mathcal O((\xi^t{-}t) \delta^{-2} )\) scaling, respectively, with the same probability.
  A small-\(t\) effect can be seen where the mean and standard deviation of the correlated samples is lower than that of the i.i.d. samples at \(t \lesssim 16\) as explained in the main text. The number of i.i.d. samples and correlated samples generated at particular \(t\)-values is shown in Fig.~\(2\).}
\label{fig:meansandstdevs}
\end{figure}

\section{Ensemble Sampling Bounds}
  \label{app:ensemblesamplingbounds}

\setcounter{lemma}{1}
  \begin{theorem}
    \label{th:ensemblebound_general}
    Given \(f_t\) correlated states that mutually satisfy \(\mathbb E(\braket{\tilde x_i}{\tilde y_{\sigma_i(j)}}) < \xi_t^{-1}\), then \(\lesssim 0.05 \xi_t \delta^{-1}\) states are sufficient in the sparsification of \(\ket{\psi}\) in order for \(\|\mathbb E\left[\frac{\ketbra{\psi}{\psi}}{\braket{\psi}{\psi}}\right] - \ketbra{\Psi}{\Psi}\| \le \delta\).
  \end{theorem}

  \begin{proof}

    Let \(\rho_1 = \mathbb E\left[\frac{\ketbra{\psi}{\psi}}{\braket{\psi}{\psi}}\right]\) and \(\rho_2 = \frac{\mathbb E (\ketbra{\psi}{\psi})}{\mathbb E (\braket{\psi}{\psi})}\).
    
    \begin{equation}
      \label{eq:rho1error_triangleineq}
    \|\rho_1 - \ketbra{\Psi}{\Psi} \|_1 \le \| \rho_1 - \rho_2 \| + \| \rho_2 - \ketbra{\Psi}{\Psi} \|_1.
  \end{equation}

  \begin{align}
    \mu \rho_2 =& \frac{\|c\|_1^2}{k^2} \left[ \sum_{\alpha \ne \beta} \mathbb E (\ketbra{\omega_\alpha}{\omega_\beta}) + \sum_\alpha^k \mathbb E ( \ketbra{\omega_\alpha}{\omega_\alpha}) \right] \\
    =& \frac{\|c\|_1^2}{k^2} \left[ \sum_\alpha^k \sum_{\beta \ne \alpha}^{k-1-f_t} \mathbb E(\ketbra{\omega_\alpha}{\omega_\beta}) \right.\\
                & \left. + \sum_{\alpha}^k \sum_{\beta \ne \alpha}^{f_t} \mathbb E (\ketbra{\omega_{\alpha}}{\omega_{f_\alpha(\beta)}}) + \sum_\alpha^k \mathbb E(\ketbra{\omega_\alpha}{\omega_\alpha})\right] \nonumber\\
    \equiv& \frac{\|c\|_1^2}{k^2} \left[ \sum_\alpha^k \sum_{\beta \ne \alpha}^{k-1-f_t} \frac{\ketbra{\Psi}{\Psi}}{\|c\|_1^2} + \sum_{\alpha}^k \sum_{\beta \ne \alpha}^{f_t} \tilde \theta + \sum_\alpha^k \sigma \right] \\
    =& \frac{\|c\|_1^2}{k^2} \left[ \frac{k (k-1-f_t)}{\|c\|_1^2} \frac{\ketbra{\Psi}{\Psi}}{\|c\|_1^2} + k f_t \tilde \theta + k \sigma \right] \\
    =& (1 - k^{-1}(1+f_t)) \ketbra{\Psi}{\Psi} + k^{-1} \|c\|_1^2 f_t \tilde \theta\\
                & + k^{-1} \|c\|_1^2 \sigma. \nonumber
  \end{align}

  \begin{align}
    \implies \| \rho_2 - \ketbra{\Psi}{\Psi} =& \mu^{-1} \| \mu \rho_2 - \mu \ketbra{\Psi}{\Psi} \|_1\\
    =& \mu^{-1} \| (1 - k^{-1}(1 + f_t) - \mu) \ketbra{\Psi}{\Psi} \nonumber\\
                                              & + k^{-1} \|c\|_1^2 f_t \tilde \theta + \|c\|_1^2 k^{-1} \sigma \|_1\\
    =& \mu^{-1} \| [1 - k^{-1}(1 + f_t) - \|c\|_1^2 k^{-1} \nonumber\\
                                              & - 1 + (1 + f_t) k^{-1}] \ketbra{\Psi}{\Psi}\\
                                              & + k^{-1} \|c\|_1^2 f_t \tilde \theta + \|c\|_1^2 k^{-1} \sigma \|_1 \nonumber\\
    \le& \mu^{-1} \|k^{-1} \|c\|_1^2 f_t \tilde \theta + \|c\|_1^2 k^{-1} \sigma \\
                                              & - \|c\|_1^2 k^{-1} \ketbra{\Psi}{\Psi}\|_1 \nonumber\\
    \le& \mu^{-1} \|c\|^2_1 \|1 - \ketbra{\Psi}{\Psi}\|_1 k^{-1} \nonumber\\
    \label{eq:tighterstateoftheartprefactor}
    \le& 1.
  \end{align}
  for \(t \gg 1\). We used \(\mu = \|c\|_1^2k^{-1} + 1 - (1+f_t)k^{-1}\) in the second step, the triangle inequality, \(\|\tilde \theta\| \le \mathbb E | \Tr \ketbra{\omega_\alpha}{\omega_{f_\alpha(\beta)}} | < \|c\|_1^{-2}\) in the second-last step.

  In the original proof~\cite{Seddon20}, the property \( \mu^{-1} \le 1\) and \(\|\ketbra{\Psi}{\Psi}\|_1 \le 1\) is used, along with the triangle inequality, to simplify the second-last equation to \(\le 2 \|c\|_1^2 k^{-1}\). However, more accurately \(\mu \ge 2\|c\|_1^2 k^{-1} \Leftrightarrow \mu^{-1} \le (2 \|c\|_1^2 k^{-1})^{-1} \) and so the final equation is more tightly lowered bounded by \(\le 1\). This tightens the state-of-the-art asymptotic prefactor from \((2+\sqrt{2})\) to \(\sqrt{2}\) even with only i.i.d. sampling.

  \begin{align}
    \|\rho_1 - \rho_2\|_1 \le \sqrt{\text{Var} \braket{\psi}{\psi}}.
  \end{align}

  \begin{align}
    \mu = \mathbb E (\braket{\psi}{\psi}) = \|c\|_1^2 k^{-1} + \|c\|_1^2 \mathbb E (B) k^{-2},
  \end{align}
  where \(B = \sum_{\alpha \ne \beta} \braket{\omega_\alpha}{\omega_\beta}\).

  \begin{align}
    \implies& 1 - \gamma k^{-1} = \|c\|_1^2 \mathbb E(B) k^{-2}\\
    \Leftrightarrow& \quad \mathbb E(B) = k^2 \|c\|_1^{-2}(1 - \gamma k^{-1}).
  \end{align}

  \begin{align}
    \text{Var} \braket{\psi}{\psi} =& \|c\|_1^4 k^{-4}( \mathbb E(B^2) - \mathbb E(B)^2 )\\
    =& \|c\|_1^4 k^{-4} ( \mathbb E(B^2) - k^4 \|c\|_1^{-4} (1 - \gamma k^{-1})^2 )\\
    =& \|c\|_1^4 k^{-4} \mathbb E(B^2) - 1 + 2 \gamma k^{-1} - \gamma^2 k^{-2}.
  \end{align}

  \begin{align}
    B^2 =& \sum_{\alpha \ne \beta} \braket{\omega_\alpha}{\omega_\beta} \sum_{\gamma \ne \delta} \braket{\omega_\gamma}{\omega_\delta}\\
    =& \sum_{\substack{(\alpha, \beta, \gamma, \delta)\\\text{all distinct}}} \braket{\omega_\alpha}{\omega_\beta}\braket{\omega_\gamma}{\omega_\delta} + B'.
  \end{align}

  \begin{align}
    &\sum_{\substack{(\alpha, \beta, \gamma, \delta)\\\text{all distinct}}} \braket{\omega_\alpha}{\omega_\beta}\braket{\omega_\gamma}{\omega_\delta} \nonumber\\
    =& \sum_{\alpha = 1}^k \sum_{\beta \ne \alpha}^{k-f_t-1}\sum_{\gamma \ne \alpha, \beta}^{k-2f_t-2}\sum_{\delta \ne \alpha, \beta, \gamma}^{k-3f_t-3} \mathbb E (\bra{\omega_\alpha}) \mathbb E (\ket{\omega_\beta}) \mathbb E (\bra{\omega_\gamma}) \mathbb E(\ket{\omega_\delta}) \nonumber\\
    & + \sum_{\alpha=1}^k \sum_{\beta \ne \alpha}^{f_t} \sum_{\gamma \ne \alpha, \beta}^{k-f_t-1} \sum_{\delta \ne \alpha, \beta, \gamma}^{k-2f_t-2} \mathbb E(\braket{\omega_\alpha}{\omega_{f_\alpha(\beta)}}) \mathbb E(\bra{\omega_\gamma}) \mathbb E(\ket{\omega_\delta}) \nonumber\\
    & + \sum_{\alpha=1}^k \sum_{\beta \ne \alpha}^{k-f_t-1} \sum_{\gamma\ne\alpha,\beta}^{f_t} \sum_{\delta\ne\alpha,\beta,\gamma}^{k-2f_t-2} \mathbb E(\bra{\omega_\alpha}) \mathbb E(\ketbra{\omega_\beta}{\omega_{f_\beta(\gamma)}}) \mathbb E(\ket{\omega_\delta}) \nonumber\\
    & + \sum_{\alpha=1}^k \sum_{\beta\ne\alpha}^{k-f_t-1} \sum_{\gamma\ne\alpha,\beta}^{k-2f_t-2} \sum_{\delta\ne\alpha,\beta,\gamma}^{f_t} \mathbb E(\bra{\omega_\alpha}) \mathbb E(\ket{\omega_\beta}) \mathbb E(\braket{\omega_\gamma}{\omega_{f_\gamma(\delta)}}) \nonumber\\
    & + \sum_{\alpha=1}^k \sum_{\beta\ne\alpha}^{f_t} \sum_{\gamma\ne\alpha,\beta}^{k-f_t-2} \sum_{\delta\ne\alpha,\beta,\gamma}^{f_t} \mathbb E(\braket{\omega_\alpha}{\omega_{f_\alpha(\beta)}}) \mathbb E(\braket{\omega_\gamma}{\omega_{f_\gamma(\delta)}}) \nonumber\\
    & + \sum_{\alpha=1}^k \sum_{\beta\ne\alpha}^{k-f_t-1} \sum_{\gamma \ne \alpha,\beta}^{f_t} \sum_{\delta \ne \alpha, \beta, \gamma}^{f_t} \mathbb E (\braket{\omega_\alpha}{\omega_\beta} \braket{\omega_{f_\alpha(\delta)}}{\omega_{f_\beta(\delta)}}) \nonumber\\
    & + \sum_{\alpha=1}^k \sum_{\beta\ne\alpha}^{k-f_t-1} \sum_{\gamma\ne\alpha,\beta}^{f_t} \sum_{\delta\ne\alpha,\beta,\gamma}^{f_t} \mathbb E( \braket{\omega_\alpha}{\omega_\beta} \braket{\omega_{f_\beta(\gamma)}}{\omega_{f_\alpha(\delta)}}) \nonumber\\
    & + \sum_{\alpha=1}^k \sum_{\beta\ne\alpha}^{f_t} \sum_{\gamma\ne\alpha,\beta}^{f_t-1} \sum_{\delta\ne\alpha,\beta,\gamma}^{k-f_t-1} \mathbb E(\braket{\omega_\alpha}{\omega_{f_\alpha(\beta)}} \bra{\omega_{f_\beta(\gamma)}}) \mathbb E(\ket{\omega_\delta}) \nonumber\\
    & + \sum_{\alpha=1}^k \sum_{\beta\ne\alpha}^{f_t} \sum_{\gamma\ne\alpha,\beta}^{k-f_t-1} \sum_{\delta\ne\alpha,\beta,\gamma}^{f_t-1} \mathbb E (\braket{\omega_\alpha}{\omega_{f_\alpha(\beta)}} \mathbb E(\bra{\omega_\gamma}) \ket{\omega_{f_\beta(\delta)}} ) \nonumber\\
    & + \sum_{\alpha=1}^k \sum_{\beta\ne\alpha}^{k-f_t-1} \sum_{\gamma\ne\alpha,\beta}^{f_t} \sum_{\delta\ne\alpha,\beta,\gamma}^{f_t-1} \mathbb E (\bra{\omega_\alpha} \mathbb E(\ket{\omega_\beta}) \braket{\omega_{f_\alpha(\gamma)}}{\omega_{f_\gamma(\delta)}}) \nonumber\\
    & + \sum_{\alpha=1}^k \sum_{\beta\ne\alpha}^{k-f_t-1} \sum_{\gamma\ne\alpha,\beta}^{f_t} \sum_{\delta\ne\alpha,\beta,\gamma}^{f_t-1} \mathbb E(\bra{\omega_\alpha}) \mathbb E(\ket{\omega_\beta}\braket{\omega_{f_\beta(\gamma)}}{\omega_{f_\gamma(\delta)}}) \nonumber\\
    & + \sum_{\alpha=1}^k \sum_{\beta\ne\alpha}^{f_t} \sum_{\gamma\ne\alpha,\beta}^{f_t-1} \sum_{\delta\ne\alpha,\beta,\gamma}^{f_t-2} \mathbb E (\braket{\omega_\alpha}{\omega_{f_\alpha(\beta)}} \braket{\omega_{f_\beta(\gamma)}}{\omega_{f_\gamma(\delta)}})\\
    \le& \|c\|_1^{-4} k (k-1-f_t) (k-2-2f_t)(k-3-3f_t) \nonumber\\
    & + k f_t \theta \|c\|_1^{-2} (k-1-f_t)(k-2-2f_t) \nonumber\\
    & + k (k-1-f_t)(k-2-2f_t) \expval{\Psi}{\tilde \theta}{\Psi} \nonumber\\
    & + \|c\|_1^{-2} k(k-1-f_t)(k-2-2f_t)f_t\theta \nonumber\\
    & + k f_t \theta(k-2-2f_t)f_t\theta \nonumber\\
    & + \|c\|_1^{-4} k(k-1-f_t) f_t^2 \nonumber\\
    & + \|c\|_1^{-4} k(k-1-f_t) f_t^2 \nonumber\\
    & + k f_t |\theta|^2 (f_t-1)(k-1-f_t) \nonumber\\
    & + k f_t |\theta|^2 (f_t-1)(k-1-f_t) \nonumber\\
    & + \|c\|_1^{-2} k(k-1-f_t)f_t(f_t-1)|\theta| \nonumber\\
    & + \|c\|_1^{-2} k(k-1-f_t)f_t(f_t-1)|\theta| \nonumber\\
    & + k f_t |\theta|^2(f_t-1)(f_t-2),
  \end{align}
  \begin{align}
    &\sum_{\substack{(\alpha, \beta, \gamma, \delta)\\\text{all distinct}}} \braket{\omega_\alpha}{\omega_\beta}\braket{\omega_\gamma}{\omega_\delta} \nonumber\\
    \le& k^4 \|c\|_1^{-4}\\
    & + k^3 (-6 \|c\|_1^{-4} -6 \|c\|_1^{-4} f_t + f_t \theta \|c\|_1^{-2} + f_t \|c\|_1^{-2} |\theta| \nonumber\\
    & \quad \qquad + \|c\|_1^{-2} f_t \theta) \nonumber\\
    & + k^2 (-3 f_t \theta \|c\|_1^{-2} - 3 f_t^2\theta\|c\|_1^{-2}-3 f_t \theta \|c\|_1^{-2} - 3 f_t^2\theta\|c\|_1^{-2} \nonumber\\
    & \quad \qquad -3 f_t \theta \|c\|_1^{-2} - 3 f_t^2\theta\|c\|_1^{-2} +  f_t^2 \theta^2 + \|c\|_1^{-4}  f_t^2 \nonumber\\
    & \quad \qquad + \|c\|_1^{-4}  f_t^2 + f_t (f_t-1) |\theta|^2 + f_t(f_t-1)|\theta|^2 \nonumber\\
    & \quad \qquad + \|c\|_1^{-2} f_t(f_t-1) |\theta| + \|c\|_1^{-2} f_t(f_t-1) |\theta|) \nonumber\\
    & + k ( -\|c\|^{-4} (1+f_t)(2+2f_t)(3+3f_t) \nonumber\\
    & \quad \qquad + f_t \theta \|c\|_1^{-2} (1+f_t)(2+2f_t) \nonumber\\
    & \quad \qquad + (1+f_t)(2+2f_t)f_t\|c\|_1^{-2}|\theta| \nonumber\\
    & \quad \qquad + \|c\|_1^{-2} f_t \theta (1+f_t)(2+2f_t) \nonumber\\
    & \quad \qquad + f_t \theta (-2 -2f_t) f_t \theta \nonumber\\
    & \quad \qquad + \|c\|_1^{-4} (-1-f_t) f_t^2 \nonumber\\
    & \quad \qquad + \|c\|_1^{-4} (-1-f_t) f_t^2 \nonumber\\
    & \quad \qquad + f_t |\theta| (f_t-1)(-1-f_t)|\theta| \nonumber\\
    & \quad \qquad + f_t |\theta| (f_t-1)(-1-f_t)|\theta| \nonumber\\
    & \quad \qquad + \|c\|^{-2} (-1-f_t)f_t(f_t-1)|\theta| \nonumber\\
    & \quad \qquad + \|c\|^{-2} (-1-f_t)f_t(f_t-1)|\theta| \nonumber\\
    & \quad \qquad + f_t |\theta| (f_t-1)(f_t-2) |\theta|), \nonumber
  \end{align}
  where
  \begin{equation}
    \theta \equiv \mathbb E(\braket{\omega_\alpha}{\omega_{f_\alpha(\beta)}}),
  \end{equation}
  \begin{equation}
    \sigma \equiv \mathbb E (\ketbra{\omega_\alpha}{\omega_\alpha}),
  \end{equation}
  and
  \begin{equation}
    \tilde \sigma \equiv \sum_\delta^{f_t} \mathbb E (\ketbra{\omega_\gamma}{\omega_\gamma(\delta)}).
  \end{equation}

  We used the bounds
  \begin{align}
    |\mathbb E_{\delta \ne \gamma} (\braket{\omega_\gamma}{\omega_\delta(\epsilon)})| \le& \frac{|\mathbb E( \braket{\omega_\delta}{\omega_{f_\delta(\epsilon)}})|}{|\mathbb E( \braket{\omega_\delta}{\omega_\gamma})|}\\
    \le& |\mathbb E(\braket{\omega_\delta}{\omega_{f_\delta(\epsilon)}})|\\
    =& \theta,
  \end{align}
  and from the Cauchy-Schwartz inequality \(|\mathbb E (XY)|^2 {\le} |\mathbb E (X^2)| |\mathbb E(Y^2)| {\le} \mathbb E(|X|)^2 \mathbb E(|Y|)^2\),
  \begin{align}
    & |\mathbb E_{\alpha\ne\beta\ne\gamma\ne\delta}(\braket{\omega_\alpha}{\omega_\beta}\braket{\omega_{f_\alpha(\gamma)}}{\omega_{f_\beta(\delta)}})| \nonumber\\
    \le& \sqrt{\mathbb E|\braket{\omega_\alpha}{\omega_\beta}|^2 \mathbb E |\braket{\omega_{f_\alpha(\gamma)}}{\omega_{f_\beta(\delta)}}|^2}\\
    \le& \sqrt{\|c\|_1^{-4}\|c\|_1^{-4}}\\
    =& \|c\|_1^{-4},
  \end{align}
  \begin{align}
    & |\mathbb E(\braket{\omega_\alpha}{\omega_{f_\alpha(\beta)}}\bra{\omega_{f_\beta(\gamma)}}) \mathbb E(\ket{\omega_\delta})| \nonumber\\
    \le& \sqrt{\mathbb E|\braket{\omega_\alpha}{\omega_{f_\alpha(\beta)}}|^2 \mathbb E|\braket{\omega_{f_\beta(\gamma)}}{\omega_\delta}|^2}\\
    =& \sqrt{|\theta|^2|\theta|^2}\\
    =& |\theta|^2,
  \end{align}
  \begin{align}
    |\mathbb E(\braket{\omega_{f_\alpha(\beta)}}{\omega_{f_\beta(\gamma)}})| = |\theta|,
  \end{align}
  \begin{align}
    & |\mathbb E(\braket{\omega_\alpha}{\omega_\beta}\braket{\omega_{f_\alpha(\gamma)}}{\omega_{f_\gamma(\delta)}})| \nonumber\\
    \le& \sqrt{\mathbb E|\braket{\omega_\alpha}{\omega_\beta}|^2 \mathbb E |\braket{\omega_{f_\alpha(\gamma)}}{\omega_{f_\gamma(\delta)}}|^2 }\\
    =& \sqrt{ \|c\|_1^{-4} |\theta|^2} \\
    =& \|c\|_1^{-2} |\theta|,
  \end{align}
  \begin{align}
    & |\mathbb E\braket{\omega_\alpha}{\omega_{f_\alpha(\beta)}}\braket{\omega_{f_\beta(\gamma)}}{\omega_{f_\gamma(\delta)}}| \nonumber\\
    \le& \sqrt{|\mathbb E (\braket{\omega_\alpha}{\omega_{f_\alpha(\beta)}}) |^2 |\mathbb E (\braket{\omega_{f_\beta(\gamma)}}{\omega_{f_\gamma(\delta)}}|^2}\\
    =& |\theta|^2.
  \end{align}

  Note that
  \begin{equation}
    \expval{\Psi}{\tilde \sigma}{\Psi} \le f_t \|c\|_1^{-2} |\theta|.
  \end{equation}

  \begin{align}
    B' =& B_{\lambda=\alpha} + B_{\mu=\alpha} + B_{\lambda=\beta} + B_{\mu=\beta}\\
        & + B_{\lambda=\alpha; \mu=\beta} + B_{\mu=\alpha; \lambda=\beta}\nonumber\\
    =& 2 \Re [ \mathbb E(B_{\lambda=\alpha}) + \mathbb E(B_{\mu=\beta})] \\
        & + \mathbb E(B_{\lambda=\alpha; \mu=\beta} + B_{\mu=\alpha; \lambda=\beta}) \nonumber.
  \end{align}

  \begin{align}
    \mathbb E(B_{\lambda=\beta}) =& \sum_{\alpha\ne\beta\ne\gamma} \mathbb E(\braket{\omega_\alpha}{\omega_\beta} \braket{\omega_\beta}{\omega_\gamma})\\
    =& \sum_{\alpha=1}^k \sum_{\beta\ne\alpha}^{k-f_t-1} \sum_{\gamma\ne\alpha\ne\beta}^{k-f_t-2} \mathbb E(\bra{\omega_\alpha})\mathbb E(\ketbra{\omega_\beta}{\omega_\beta})\mathbb E(\ket{\omega_\gamma}) \nonumber\\
    & + \sum_{\alpha=1}^k \sum_{\beta\ne\alpha}^{f_t} \sum_{\gamma\ne\alpha,\beta}^{k-f_t-1} \mathbb E( \braket{\omega_\alpha}{\omega_{f_\alpha(\beta)}} \bra{\omega_{f_\alpha(\beta)}}) \mathbb E(\ket{\omega_\gamma}) \nonumber\\
    & + \sum_{\alpha=1}^k \sum_{\beta\ne\alpha}^{k-f_t-1} \sum_{\gamma\ne\alpha,\beta}^{f_t} \mathbb E(\bra{\omega_\alpha}) \mathbb E(\ket{\omega_\beta} \expval{\omega_\beta}{\omega_{f_\beta(\gamma)}}) \nonumber\\
                                  & + \sum_{\alpha=1}^k \sum_{\beta\ne\alpha}^{f_t} \sum_{\gamma\ne\alpha,\beta}^{f_t-1} \mathbb E( \braket{\omega_\alpha}{\omega_{f_\alpha(\beta)}} \braket{\omega_{f_\alpha(\beta)}}{\omega_{f_\beta(\gamma)}} )\\
    \le& \|c\|_1^{-2} k(k-1-f_t)(k-2-f_t)\expval{\Psi}{\sigma}{\Psi} \nonumber\\
                                  & + k f_t |\theta|^2 (k-f_t-1) \nonumber\\
                                  & + \|c\|_1^{-2} k (k-f_t-1) f_t |\theta| \nonumber\\
                                  & + k f_t |\theta|^2 (f_t-1).
  \end{align}

  \begin{equation}
    |\mathbb E(B_{\mu=\beta})| = |\mathbb E(B_{\lambda=\beta})|.
  \end{equation}

  \begin{align}
    &|\mathbb E(B_{\lambda=\alpha;\mu=\beta} + B_{\mu=\alpha;\lambda=\beta})| \nonumber\\
    \le& 2 \sum_{\alpha\ne\beta} \mathbb E |\braket{\omega_\alpha}{\omega_\beta}\braket{\omega_\beta}{\omega_\alpha}| \\
    =& 2 \sum_{\alpha=1}^k \sum_{\beta\ne\alpha}^{k-1-f_t} \mathbb E |\braket{\omega_\alpha}{\omega_\beta}| \mathbb E |\braket{\omega_\beta}{\omega_\alpha}|\\
    & + 2 \sum_{\alpha=1}^k \sum_{\beta\ne\alpha}^{f_t} \mathbb E |\braket{\omega_\alpha}{\omega_{f_\alpha(\beta)}} \braket{\omega_{f_\alpha(\beta)}}{\omega_\alpha}|\\
    \le& 2k(k-1-f_t)\Tr (\sigma^2) + 2 k f_t |\theta|^2\\
    \le& 2k(k-1-f_t) + 2 k f_t |\theta|^2.
  \end{align}

  \begin{align}
    \implies \mathbb E (B') \le& 4 k^3 \|c\|_1^{-2} \expval{\Psi}{\sigma}{\Psi}\\
                               & + 4 k^2 ( -3 \|c\|_1^{-2} \expval{\Psi}{\sigma}{\Psi} - 2 \|c\|_1^{-2} f_t \expval{\Psi}{\sigma}{\Psi} \nonumber\\
                               & \quad \qquad f_t |\sigma|^2 + \|c\|^{-2} f_t |\theta| + 1/2) \nonumber\\
                               & + 4k (\|c\|_1^{-2}(1+f_t)(2+2f_t)\expval{\psi}{\sigma}{\psi} \nonumber\\
                               & \quad \qquad + |\theta|^2 (-f_t-1) + \|c\|_1^{-2} (-1-f_t)f_t |\theta| \nonumber\\
                               & \quad \qquad + f_t |\theta| (f_t-1) |\theta| + (-1-f_t)/2 \nonumber\\
                               & \quad \qquad + f_t |\theta|^2/2). \nonumber
  \end{align}

  Hence,
  \begin{align}
    \text{Var}(\braket{\psi}{\psi}) =& \|c\|_1^4 k^{-4} \mathbb E(B^2) - 1 + 2 \gamma k^{-1} - \gamma^2 k^{-2}\\
    \le& \|c\|_1^4 \{ \|c\|_1^{-4} + k^{-1} [4 \|c\|_1^{-2} \expval{\Psi}{\sigma}{\Psi} \nonumber\\
                                         & \quad \quad - 6 \|c\|_1^{-4} (f_t+1) + 3\|c\|_1^{-2}f_t|\theta| ] \nonumber\\
                                         & + k^{-2} [-\|c\|_1^{-2} (3+2f_t)\expval{\Psi}{\sigma}{\Psi} + f_t|\theta|^2 \nonumber\\
                                         & \quad \qquad + \|c\|_1^{-2} f_t |\theta| + 2 - 9\|c\|_1^{-2}f_t\theta \nonumber\\
                                         & \quad \qquad - 9  f_t^2 \|c\|_1^{-2} \theta +  f_t^2 \theta^2 + 2 \|c\|_1^{-4}  f_t^2 \nonumber\\
                                         & \quad \qquad + 2 f_t(f_t-1)|\theta|^2 \nonumber\\
                                         & \quad \qquad + 2\|c\|_1^{-2}f_t(f_t-1)|\theta|]\} \nonumber\\
                                     & - 1 + 2 k^{-1} \gamma - k^{-2} \gamma^2 \nonumber\\
                                     &+ k^{-3} \|c\|_1^4 [\|c\|_1^{-4} (1+f_t)(2+2f_t)(3+3f_t) \nonumber\\
                                          & \quad \qquad -\|c\|_1^{-4}f_t^2(1+f_t) \nonumber\\
                                            & \quad \qquad -\|c\|_1^{-4}f_t^2(1+f_t) \nonumber\\
                                            & \quad \qquad + 4\|c\|_1^{-2}(1+f_t)(2+2f_t)\expval{\psi}{\sigma}{\psi} \nonumber\\
                                            & \quad \qquad - 2(1+f_t)] \end{align}
  \begin{equation}
    \gamma = 1 + f_t - \|c\|_1^2 \theta f_t.
  \end{equation}

  \begin{align}
    \implies \text{Var}(\braket{\psi}{\psi}) \le& \|c\|_1^4\{ k^{-1} [4 \|c\|_1^{-2} \expval{\Psi}{\sigma}{\Psi} \nonumber\\
                                                    & - 6\|c\|_1^{-4}(f_t+1) + 3 \|c\|_1^{-2} f_t |\theta|] \nonumber\\
                                                    & + k^{-2} [ -\|c\|_1^{-2}(3+2f_t)\expval{\Psi}{\sigma}{\Psi} \nonumber\\
                                                    & + f_t|\theta|^2 + \|c\|_1^{-2} f_t|\theta| + 2 - 9 \|c\|_1^{-2}f_t\theta \nonumber\\
                                                    & - 9 f_t^2\|c\|_1^{-2}\theta +  f_t^2 \theta^2 + 2\|c\|_1^{-4} f_t^2 \nonumber\\
                                                    & + 2f_t(f_t-1)|\theta|^2 + 2\|c\|_1^{-2}f_t(f_t-1)|\theta|]\} \nonumber\\
                                                    & + 2 k^{-1}(1+f_t-\|c\|_1^2\theta f_t) \nonumber\\
                                                    & - k^{-2}(1 + 2f_t +  f_t^2 + \|c\|_1^4 \theta^2  f_t^2 \nonumber\\
                                                    & - 2\|c\|_1^2\theta f_t - 2\|c\|_1^2\theta  f_t^2) \nonumber\\
                                                & + k^{-3} \|c\|_1^4 [\|c\|_1^{-4} (1+f_t)(2+2f_t)(3+3f_t) \nonumber\\
                                          & \quad \qquad -\|c\|_1^{-4}f_t^2(1+f_t) \nonumber\\
                                            & \quad \qquad -\|c\|_1^{-4}f_t^2(1+f_t) \nonumber\\
                                            & \quad \qquad + 4\|c\|_1^{-2}(1+f_t)(2+2f_t)\expval{\psi}{\sigma}{\psi} \nonumber\\
                                                & \quad \qquad - 2(1+f_t)]
  \end{align}

  We now take every factor and impose the bound, \(0 \le f_t \theta \le \|c\|_1^2\theta \ll 1\), as \(t \rightarrow \infty\), since they converge to \(0\) exponentially w.r.t. \(t\). We also assume \(f_t \gg 1\).

  \begin{align}
    \implies \text{Var}(\braket{\psi}{\psi}) \le& \|c\|_1^4\{k^{-1}[4\|c\|_1^{-2}\expval{\Psi}{\sigma}{\Psi} \nonumber\\
                                                    & \quad \quad - 6 \|c\|_1^{-4} (f_t+1) + 3\|c\|_1^{-2}f_t|\theta|] \nonumber\\
                                                    & +k^{-2} [-\|c\|_1^{-2}(3+2f_t)\expval{\Psi}{\sigma}{\Psi}+f_t|\theta|^2\nonumber\\
                                                    & -8\|c\|_1^{-2}f_t|\theta| + 2 - 9 f_t^2\|c\|_1^{-2}\theta +  f_t^2 \theta^2 \nonumber\\
                                                    & + 2\|c\|_1^{-4} f_t^2 + 2  f_t^2|\theta|^2 + 2  f_t^2\|c\|_1^{-2}\theta] \nonumber\\
                                                    & + 2k^{-1}(1+ f_t) \nonumber\\
                                                & - k^{-2}( f_t^2 + \|c\|_1^4\theta^2 f_t^2) \nonumber\\
                                                & + k^{-3} [\|c\|_1^{-4} 6 f_t^3-2 \|c\|_1^{-4}f_t^3 \nonumber\\
                                            & \quad \qquad -2 \|c\|_1^{-4}f_t^2 + 4\|c\|_1^{-2}2f^2_t\expval{\psi}{\sigma}{\psi} \nonumber\\
                                                & \quad \qquad - 2 f_t]\}\\
    \le& \|c\|_1^4\{k^{-1}[4\|c\|_1^{-2}\expval{\Psi}{\sigma}{\Psi} - 4 \|c\|_1^{-4} \nonumber\\
                                                    & \qquad \qquad -4\|c\|_1^{-4}f_t + 3\|c\|_1^{-2}f_t|\theta|] \nonumber\\
                                                & \qquad + k^{-2}[2 + 2\|c\|_1^{-4}f_t^2] \nonumber\\
                                                & \quad + k^{-3} (4 \|c\|_1^{-4} f_t^3 + 8 \|c\|_1^{-2} f_t^2 \expval{\psi}{\sigma}{\psi} \nonumber\\
                                                & \qquad \qquad- 2 \|c\|_1^{-4}f_t^2) \}\\
    \le& k^{-1}[4 \|c\|_1^2\expval{\Psi}{\sigma}{\Psi} - 4 - 4f_t + 3\|c\|_1^2f_t|\theta|] \nonumber\\
                                                & \quad + k^{-2}[2 \|c\|_1^4 + 2 f_t^2] \nonumber\\
                                                & \quad + k^{-3}[4 f_t^3 + 8\|c\|_1^{2} f_t^2 \expval{\psi}{\sigma}{\psi} - 2 f_t^2] \\
    \le& k^{-1} 4(C-1-f_t) + k^{-2} 2 \|c\|_1^4 \nonumber\\
    \label{eq:variance_final}
                                                & + k^{-3}(4 f_t^3 + 2(4 C- 1)f_t^2),
  \end{align}
  for \(C \equiv \|c\|_1^2 \expval{\Psi}{\sigma}{\Psi}\).

  Therefore, plugging Eq.~\ref{eq:tighterstateoftheartprefactor} and Eq.~\ref{eq:variance_final} into Eq.~\ref{eq:rho1error_triangleineq}, we find
  \begin{align}
    \|\rho_1 - \ketbra{\Psi}{\Psi}\|_1 \le& \sqrt{4k^{-1}(C-1-f_t)+2k^{-2}\|c\|_1^4} \nonumber\\
                                          & \quad \overline{+ k^{-3}(4 f_t^3 + 2(4 C- 1)f_t^2)} + 1.
  \end{align}

  For Clifford magic states, \(C=1\)~\cite{Seddon20}. In general, \(C\) is a low exponential with \(t\). Therefore, we make the assumption that \(f_t \gg C\) as \(t \rightarrow \infty\) and so \(f_t^3 \gg C f_t^2\). We also can drop the factor of \(1\) as \(t \rightarrow \infty\).
  \begin{align}
    \implies& \|\rho_1 - \ketbra{\Psi}{\Psi}\|_1 \substack{\le\\t \rightarrow \infty} \nonumber\\
    &\sqrt{-4k^{-1}f_t+k^{-2}(2\|c\|_1^4+f_t^2) + k^{-3}f_t^3}.
  \end{align}

  To get \(\|\rho_1 - \ketbra{\Psi}{\Psi}\|_1 \le \delta\) it follows that
  \begin{widetext}
  \begin{align}
    \label{eq:kscaling}
    k \ge& \frac{2^{2/3}}{6 \delta^2} \sqrt[3]{27 \delta ^4 f_t^3-36 \delta ^2 f_t^3-128 f_t^3+\sqrt{\left(\left(27 \delta ^4-36 \delta ^2-128\right) f_t^3-72 \delta ^2 f_t \xi ^{2 t}\right)^2-4 \left(\left(3 \delta ^2+16\right) f_t^2+6 \delta ^2 \xi ^{2 t}\right)^3}-72 \delta ^2 f_t \xi ^{2 t}} \nonumber\\
         & +\frac{1}{6 \delta ^2}\frac{2 \sqrt[3]{2} \left(\left(3 \delta ^2+16\right) f_t^2+6 \delta ^2 \xi ^{2 t}\right)}{\sqrt[3]{27 \delta ^4 f_t^3-36 \delta ^2 f_t^3-128 f_t^3+\sqrt{\left(\left(27 \delta ^4-36 \delta ^2-128\right) f_t^3-72 \delta ^2 f_t \xi ^{2 t}\right)^2-4 \left(\left(3 \delta ^2+16\right) f_t^2+6 \delta ^2 \xi ^{2 t}\right)^3}-72 \delta ^2 f_t \xi ^{2 t}}}\nonumber\\
         & -\frac{4}{3} \frac{f_t}{\delta^2}.
  \end{align}

  Plugging in \(f_t = \beta \xi^t / \delta\) produces
  \begin{align}
    \label{eq:kscaling_2}
    k \ge& \frac{2^{2/3} \sqrt[3]{\frac{\beta  \left(\beta ^2 \left(27 \delta ^4-36 \delta ^2-128\right)-72 \delta ^4\right) \xi ^{3 t}}{\delta ^3}+3 \sqrt{3} \sqrt{-\frac{\left(\beta ^6 \left(-27 \delta ^4+76 \delta ^2+272\right)+8 \beta ^4 \left(21 \delta ^2+8\right) \delta ^2+16 \beta ^2 \left(3 \delta ^2+4\right) \delta ^4+32 \delta ^8\right) \xi ^{6 t}}{\delta ^2}}}}{6 \delta ^2}\\
    & +\frac{\frac{2 \sqrt[3]{2} \left(\beta ^2 \left(3 \delta ^2+16\right)+6 \delta ^4\right) \xi ^{2 t}}{\delta ^2 \sqrt[3]{\frac{\beta  \left(\beta ^2 \left(27 \delta ^4-36 \delta ^2-128\right)-72 \delta ^4\right) \xi ^{3 t}}{\delta ^3}+3 \sqrt{3} \sqrt{-\frac{\left(\beta ^6 \left(-27 \delta ^4+76 \delta ^2+272\right)+8 \beta ^4 \left(21 \delta ^2+8\right) \delta ^2+16 \beta ^2 \left(3 \delta ^2+4\right) \delta ^4+32 \delta ^8\right) \xi ^{6 t}}{\delta ^2}}}}-\frac{8 \beta  \xi ^t}{\delta }}{6 \delta ^2}. \nonumber
  \end{align}

Numerical evaluation shows that for \(\delta<0.1\) the optimal \(\beta \approx 10 \delta^2\) (see Figure~\ref{fig:optimalbetanumericalruns}).

\begin{figure}[H]
\includegraphics[scale=0.38]{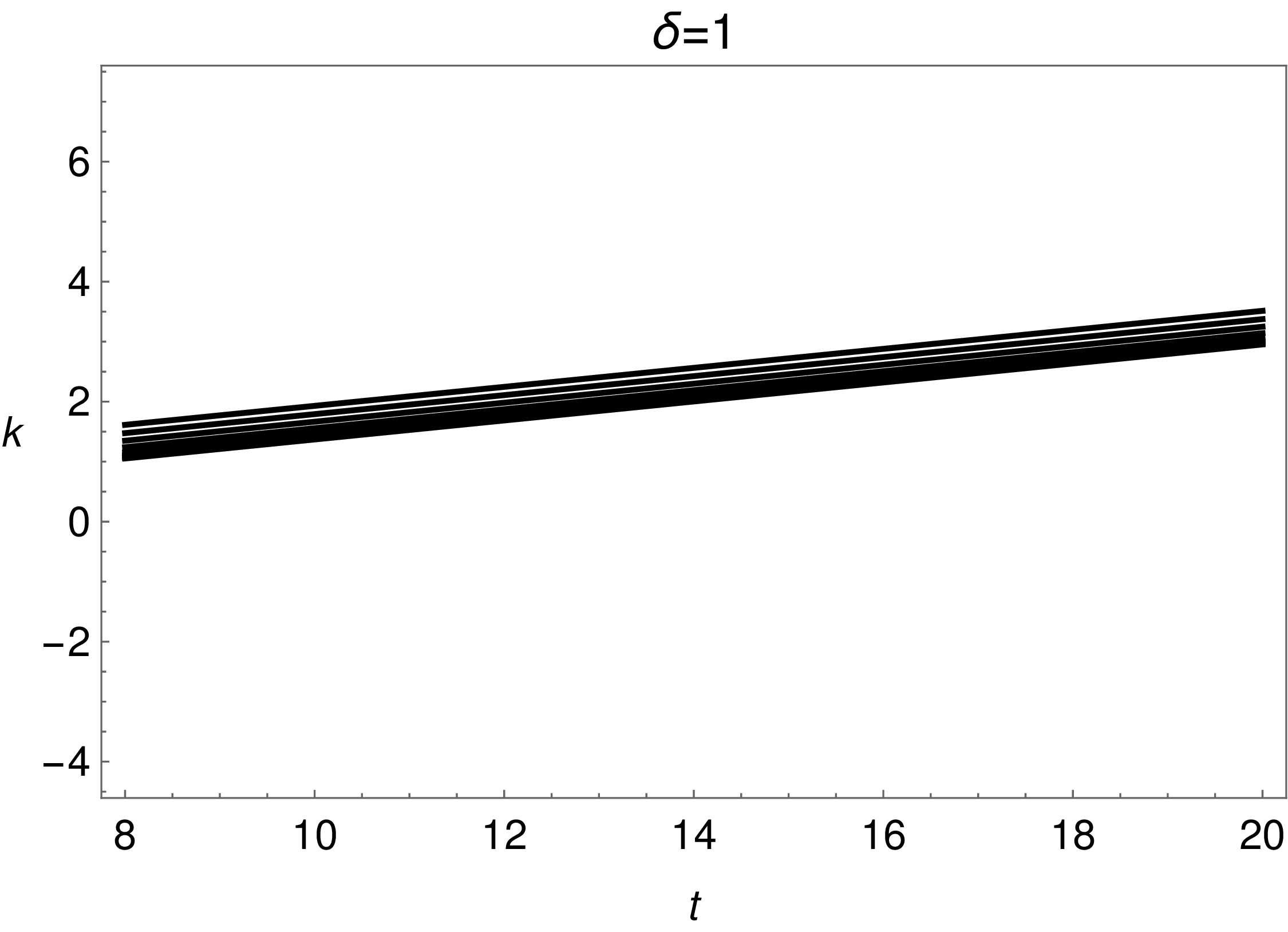}
\includegraphics[scale=0.38]{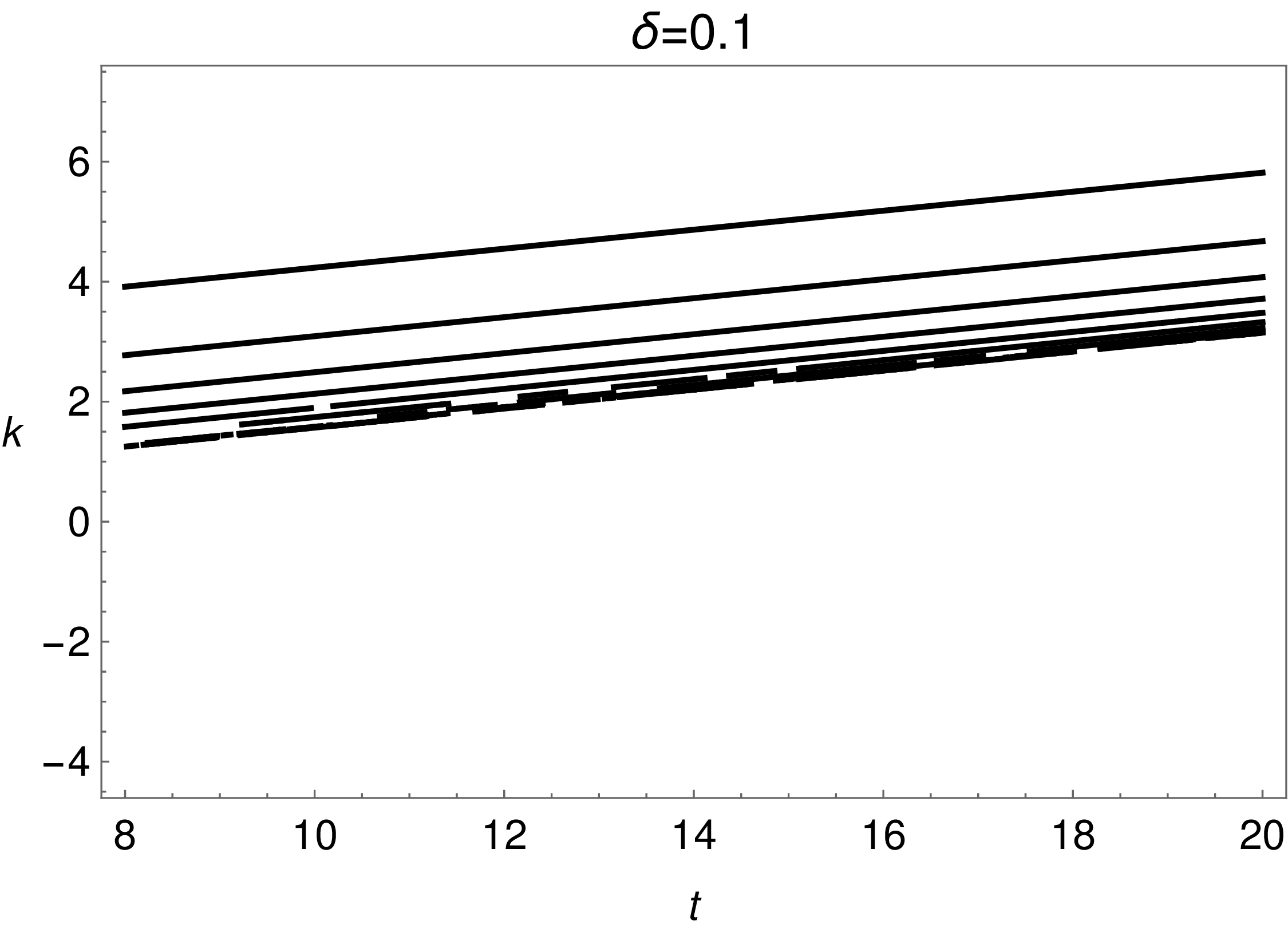}\\
\includegraphics[scale=0.38]{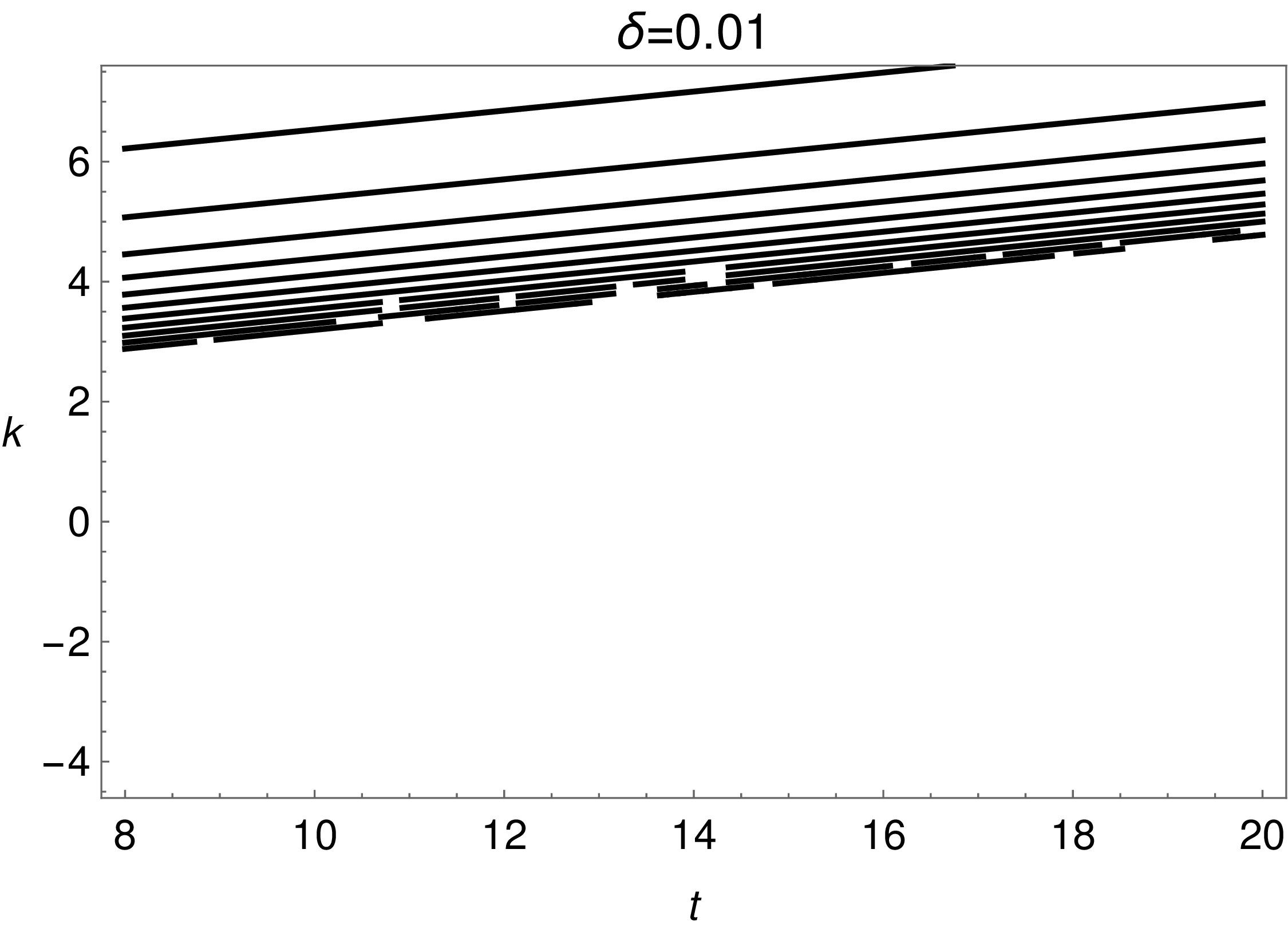}
\includegraphics[scale=0.38]{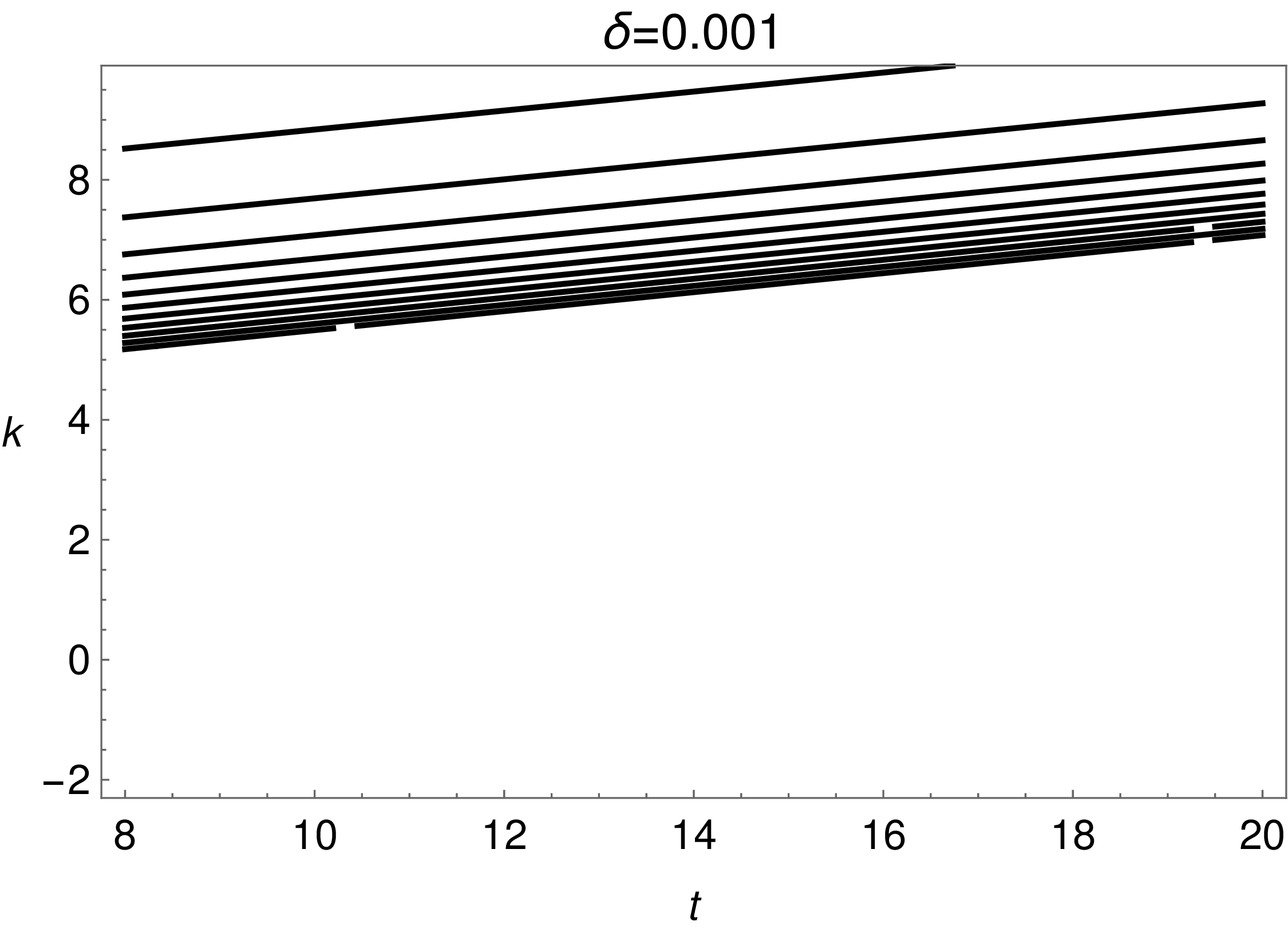}
\caption{\(k\) vs \(t\) for \(\delta=1,\,0.1,\,0.01,\,0.001\) and \(0 \le \beta \le 10 \delta^2\) in steps of \(\delta^2\). The lower the curve, the larger the \(\beta\). An decrease of \(\lesssim 0.05\) is observed for \(\delta < 0.1\) and \(t\ge 8\).}
\label{fig:optimalbetanumericalruns}
\end{figure}
\end{widetext}

Plugging this into Eq.~\ref{eq:kscaling_2} and taking the limit \(\delta \rightarrow 0\) produces

\begin{equation}
  \label{eq:kscaling_withoptimalbeta}
  k \ge (\sqrt{402}-20) \xi_t \delta^{-1} \approx 0.05 \xi^t \delta^{-1}.
\end{equation}

  We compare Eq.~\ref{eq:kscaling_withoptimalbeta} to the \(\beta=0\) leading-order term by dividing by \((\sqrt{2})\delta^{-1} \xi^t\),
  \begin{equation}
    \frac{(\sqrt{402}-20)\xi_t \delta^{-1}}{(\sqrt{2}) \xi_t \delta^{-1}} \approx 28.28.
  \end{equation}

  We compare Eq.~\ref{eq:kscaling_withoptimalbeta} to the prior state-of-the-art by dividing by \((2 + \sqrt{2})\delta^{-1} \xi^t\),
  \begin{equation}
    \frac{(\sqrt{402}-20)\xi_t \delta^{-1}}{(2+\sqrt{2}) \xi_t \delta^{-1}} \approx 68.28.
  \end{equation}

\end{proof}

\clearpage



\section{Other Comparisons}
\label{app:otherstateoftheart}

For mixed state magic states, the current state-of-the-art only scales as \(\mathcal O(\xi_t \delta^{-1})\) on average, where \(\xi_t\) is appropriately generalized to mixed states (see~\cite{Seddon20} for the algorithm on how to do this). This is because \(C \equiv \|c\|_1^2 \expval{\Psi}{\sigma}{\Psi} \ne 1\) in the proof of Theorem~\ref{th:ensemblebound_general} for non-Clifford magic states. However, for our method, this is no longer the case since the full term in the proof is \(C - 1 - f_t\) and \(f_t\) can be increased such that \(C - 1 < f_t\). This means that the correlated method not only lowers the prefactor of~\cite{Seddon20}, it also scales more favorably than the prior state-of-the-art average-case with its worst-case scaling.

Since the current state-of-the-art~\cite{Seddon20} was published, we have found two subsequent declarations in the literature of newer algorithms that outperform~\cite{Seddon20} in some regimes, though not in the worst-case.

~\cite{Pashayan21} claims to exhibit state-of-the-art weak simulation in certain regimes compared to~\cite{Seddon20}. Their favorable regime seems to be large \(p\) and low numbers of marginal probability calculations compared to the number of magic states. Unfortunately, they do not disclose the prefactors of their scaling. 

The second newer algorithm~\cite{Bravyi21} shows polynomial factor reduction in transforming approximate stabilizer decompositions to calculating outcome probabilities or outcomes in the case of shallow circuits. These improvements apply to lower the polynomial factors in our method too, and so are transferrable to improve the correlated sampling method in this regime too. Shared polynomial prefactors do not affect our relative comparisons as they get cancelled out in the ratio expressed in Eq.~\ref{eq:strongsimlowercost}. (Performing marginal calculations adds prefactors of \(t^3 w^3\) for \(t\) magic states and \(w\) marginals due to decomposing probability calculation to chain-rule marginals as described in Eq.~\(116\) of~\cite{Seddon20}. There are also factors of inverse multiplicative error squared, \(\epsilon^{-2}\), that enter when calculating outcome probabilities that are shared across weak and strong simulation.)

Factors of multiplicative error do play a role when comparing strong simulation with exact simulation. Strong simulation does not incur such an additional penalty since it already uses a full decomposition of the magic state (i.e. it has zero multiplicative error). This is the reason for the factor of \(12\) in Eq.~\ref{eq:exactsimlowercost} for strong simulation, which is obtained by budgeting sufficient additive error for the multiplicative error needed for outcome estimation~\cite{Seddon20}.

\end{document}